\pdfoutput=1
\PassOptionsToPackage{dvipsnames, table}{xcolor}

\documentclass[a4paper,UKenglish,cleveref, autoref, thm-restate]{lipics-v2021}

\usepackage{tikz}
\usepackage{nicefrac}
\usepackage[justification=centering]{caption}
\usepackage{amsmath}
\usepackage{amsthm}
\usepackage{diagbox}
\usepackage{comment}
\nolinenumbers
\usepackage{amssymb}
\usepackage{url}
\usepackage{graphicx}
\hypersetup{colorlinks,linkcolor=darkgray,citecolor=darkgray,urlcolor=gray}
\usepackage{booktabs}
\usepackage{algorithm}
\usepackage{algorithmic}
\usepackage{makecell}
\usepackage{tikz}
\usetikzlibrary{calc}
\usetikzlibrary{shapes,fit}
\usetikzlibrary{decorations.pathmorphing,decorations.pathreplacing,decorations.shapes}
\usepackage{hyperref}
\usepackage{multirow}
\usepackage{float}
\usepackage{tabularx,ragged2e}
\usepackage{blkarray}
\usepackage{subcaption}
\usepackage{thmtools}
\usepackage{thm-restate}
\usepackage[T1]{fontenc}

\newcommand{\incl}{\subseteq}

\newcommand{\emps}{\emptyset}
\newcommand{\setm}{\setminus}

\renewcommand{\d}{\delta}

\newcommand{\s}{\sigma}

\renewcommand{\t}{\tau}

\newcommand{\Oo}{\mathcal{O}}
\newcommand{\Tt}{\mathcal{T}}
\newcommand{\Uu}{\mathcal{U}}

\newcommand{\Ii}{\mathcal{I}}

\newcommand{\Hh}{\mathcal{H}}

\newcommand{\xra}[1]{\xrightarrow{#1}}
\newcommand{\I}{\operatorname{I}}

\newcommand{\Max}{\mathsf{Max}}
\newcommand{\Min}{\mathsf{Min}}

\newcommand{\prob}{\mathrm{Prob}}
\newcommand{\chance}{\mathsf{Chance}}

\newcommand{\act}{\operatorname{Act}}

\newcommand{\his}{\textit{hist}}

\newcommand{\pfr}{\textsc{pfr}}
\newcommand{\alr}{\textsc{alr}}
\newcommand{\nam}{\textsc{nam}}
\newcommand{\salr}{\textsc{s-alr}}

\newcommand{\Ee}{\mathbb{E}}

\newcommand{\Rat}{\mathbb{Q}}
\newcommand{\Real}{\mathbb{R}}

\newcommand{\NP}{\operatorname{NP}}

\newcommand{\sqsum}{\textsc{Square-Root-Sum}}

\title{Simplifying imperfect recall games\footnote{Full version of paper to appear in the proceedings of AAMAS 2025}}

\author{Hugo Gimbert}    {LaBRI, CNRS, Université de Bordeaux, France} {hugo.gimbert@labri.fr}{}{}

 \author{Soumyajit Paul} {University of Liverpool, UK\footnote{A considerable part of this work was done when this author was at LaBRI, Université de Bordeaux and at IRIF, Université Paris Cité.
}} {soumyajit.paul@liverpool.ac.uk}{}{}
 \author{B. Srivathsan}  {Chennai Mathematical Institute and CNRS, ReLaX, IRL 2000, Siruseri, India} {sri@cmi.ac.in}{}{}

\authorrunning{H.Gimbert, S.Paul and B.Srivathsan}



\funding{Soumyajit Paul - supported by the EPSRC through project EP/X03688X/1.}

\hideLIPIcs

\begin{document}
\maketitle 
\begin{abstract}

In games with imperfect recall, players may forget the sequence of decisions they made in the past. When players also forget whether they have already encountered their current decision point, they are said to be absent-minded. Solving one-player imperfect recall games is known to be $\NP$-hard, even when the players are not absent-minded. This motivates the search for polynomial-time solvable subclasses.
A special type of imperfect recall, called \emph{A-loss recall}, is amenable to efficient polynomial-time algorithms. 
In this work, we present novel techniques to simplify non-absent-minded imperfect recall games into equivalent A-loss recall games.
The first idea involves shuffling the order of actions, and leads to a new polynomial-time solvable class of imperfect recall games that extends A-loss recall. The second idea generalises the first one, by constructing a new set of action sequences which can be ``linearly combined'' to give the original game. The equivalent game has a simplified information structure, but it could be exponentially bigger in size (in accordance with the $\NP$-hardness). We present an algorithm to generate an equivalent A-loss recall game with the smallest size.

\end{abstract}


\section{Introduction}

Games play a central role in AI research. In the early $20^{th}$ century, \cite{zermelo1913} showed that perfect information games in extensive form can be solved by a bottom-up traversal of the game tree. Despite the fact that this does not readily provide efficient ways to solve large games such as Chess or Go in practice, this has indeed 
laid the foundation for the dramatic progress in the field of perfect information games, with computer programs being able to challenge human experts. Solving games becomes more intricate when the players (agents) have incomplete information about the state of the game -- Poker for instance, where a player does not know the cards of the others. One of the remarkable imperfect information games where computer programs have been able to defeat professional human players is Texas Hold'em Poker~\cite{libratus-poker,deepstack,pluribus-poker}. A main technique used in these algorithms is the abstraction of large games into smaller \emph{imperfect recall} games. 

Perfect recall is the ability of a player to remember her own actions. Poker is an imperfect information game played by several players. However, ideally one would assume that the players have a perfect recall of their actions. An imperfect recall player does not remember the sequence of her own actions. Imperfect recall allows for a structured mechanism to forget the information history and as \cite{ijcai2024p332} argues, it is particularly suited for AI agents.

From a modeling perspective, imperfect recall has been used to describe teams of agents, where each team can be represented as a single agent with imperfect recall~\cite{VONSTENGEL1997309,DBLP:conf/aaai/Celli018} or to describe agents modeling multiple nodes which do not share information between each other due to privacy reasons~\cite{DBLP:conf/aaai/Conitzer19}. Moreover~\cite{LAMBERT2019164} argues that imperfect recall is a model of bounded rationality. Given the limited memory of players, it is not realistic to assume that the players remember all their actions. We refer the reader to~\cite{ijcai2024p332} for an excellent introduction to different uses of imperfect recall.  
From a practical perspective, the most prominent use of imperfect recall is in abstracting games~\cite{PracticalUseImperfect,DBLP:conf/aaai/GanzfriedS14,DBLP:conf/atal/BrownGS15, CERMAK2020103248}. The state space generated by usual games is typically very large and abstractions are crucial for solving such games. Abstractions that preserve perfect recall force a player to distinguish the current information gained, in all later rounds, even if it is not relevant. Abstractions using players with imperfect recall have been shown to outperform those using players with perfect recall~\cite{ DBLP:conf/sara/WaughZJKSB09, johanson2013evaluating, DBLP:conf/atal/BrownGS15,DBLP:conf/ijcai/CermakBL17}.

From a complexity perspective, imperfect recall games are known to be $\NP$-hard

~\cite{KollerMegiddo::1992,Cermak::2018} even when there is a single player, whereas perfect recall games can be solved in polynomial-time~\cite{KollerMegiddo::1992,vonStengel::1996}. Recent studies have aligned the complexity of different solution concepts for imperfect recall games to the modern complexity classes~\cite{GPS20,tewolde-et-al:2023,ijcai2024p332}. The hardness of imperfect recall games has motivated the search for subclasses which are polynomial-time solvable~\cite{kline2002minimum,kaneko1995behavior}, or where algorithms similar to the perfect recall case can be applied~\cite{DBLP:conf/icml/LanctotGBB12,DBLP:conf/sigecom/KroerS16}. The class of \emph{A-loss recall}~\cite{kline2002minimum,kaneko1995behavior} is a special kind of imperfect recall, where the loss of information can be traced back to a player forgetting her own action at a point in the past -- the player remembers \emph{where} it was played, but forgets \emph{what} was played. We consider A-loss recall games to be \emph{simple} since there are polynomial-time algorithms for solving them. To the best of our knowledge, A-loss recall games are the biggest known class of imperfect recall with a polynomial-time solution. This has led to research towards finding A-loss recall abstractions~\cite{Cermak::2018}. 

\emph{Contributions.} Our broad goal in this work is to find efficient ways to solve imperfect recall games in extensive-form. We do so by simplifying them into A-loss recall games. We focus on games where the players are not absent-minded: a player is absent-minded if she even forgets whether a decision point was previously seen or not. Here are our major contributions.
\begin{enumerate}\item We first identify a class of one-player games where the player's information structure is more complex than A-loss recall, but shuffling the order of actions results in an equivalent A-loss recall game. This leads to a new $\mathsf{PTIME}$ solvable class of imperfect recall games, that extends A-loss recall (\cref{thm:1p-shuffle-ptime}, \cref{cor:effic-solv-class}, \cref{cor:2-effic-solv-class}). Furthermore, these classes themselves can be tested in $\mathsf{PTIME}$. 

\item We show that every game with \emph{non-absentminded} players can be transformed into an equivalent A-loss recall game (\cref{thm:existence-alr-span}). We present an algorithm to generate an equivalent A-loss recall game with the smallest size.
\end{enumerate}

The caveat in the second result above is that the resulting A-loss recall game could be exponentially bigger. This is expected, since solving imperfect recall games is $\NP$-hard, whereas A-loss recall games can be solved in polynomial-time. The result however shows that in order to solve imperfect recall games, one could either use a worst-case exponential-time algorithm on the original game, or apply our transformation to a worst-case exponential-sized game and run a polynomial-time algorithm on it. From a conceptual point of view, our result shows that as long as there is no absentmindedness, imperfect recall can be transformed into one where the information loss can be attributed to forgetting own actions at a past point.

\emph{Organization of the document.} Section~\ref{sec:an-example} introduces a modification of the popular matching pennies game that will be used as a running example to illustrate our results. Section~\ref{sec:background} recalls necessary preliminaries on extensive-form games. Section~\ref{sec:shuffled-loss-recall} presents the new polynomial-time class of shuffled A-loss recall. Section~\ref{sec:span} generalizes the idea of shuffling to incorporate a ``linear combination'' of action sequences, and presents the second result mentioned above. Section~\ref{sec:two-player} extends the results to the two-player setting.

\section{An example}
\label{sec:an-example}

Let us start with a one-player game called the \emph{single team matching-unmatching pennies game}, which will be used as a running example. A team of players with the same goal can be interpreted as a single player. 
In this case, the team consists of two players Alice and Bob, each possessing a coin with two sides, Head (H) and Tail (T) and each of them must choose a side for their respective coins independently. 
The game unfolds in the following manner : a fair $n$-faced die with outcomes from $\{0, \dots, n-1 \}$ is rolled; then Alice chooses a side from $\{H,T\}$, followed by Bob choosing from $\{H,T\}$. Winning or losing depends on the parity of the die outcome. If the outcome of the die is even, then they win if and only if they match their sides. If the outcome is odd, they win if and only if their sides do not match. We consider three variants depending on what Alice and Bob can observe, and model it in extensive form in \cref{fig:match-penny-3-die} for $n=3$. An informal description of the figures follows after this paragraph.
\begin{description}
  \item[I.] Both Alice and Bob observe nothing (\cref{fig:match-penny-3-die-a}).
  \item[II.] Alice can't distinguish between die outcome $2i$ and $2i+1$ for $i \geq 0$,  but Bob observes nothing (\cref{fig:match-penny-3-die-b}).
   \item[III.] Alice can't distinguish between die outcome $2i$ and $2i+1$ for $i \geq 0$, Bob only observes coin of Alice but not outcome of die (\cref{fig:match-penny-3-die-c}). 
 \end{description} 
Alice and Bob want to maximize their \emph{expected payoff}. We will see their possible strategies in Section~\ref{sec:background}.  
Later, we will see that game \textbf{I} falls under the simple class of A-loss recall. 
In Section~\ref{sec:shuffled-loss-recall} and \cref{sec:span} we will see how to simplify games \textbf{II} and \textbf{III} respectively. 
\input{Figures/fig-matching-penny-3-faced-die}

Before we delve into the background and results, here is a description of the extensive-form model. 
The root node, marked with a triangle, is the event of rolling the die. The triangle nodes are called $\chance$ nodes, and the
edges out of them associate probabilities to each of the outcomes. For this game, the distribution is uniform. The circle nodes denote decision nodes of the team. The nodes in the second level (root being the first level) belong to Alice whereas the nodes in the third level belong to Bob. The actions labelled in edges out of these nodes denote the actions available to the corresponding players. 
A leaf node indicates an end state, and a path from root to leaf denotes
a play from start to end. The number associated with a leaf gives the
payoff that the team receives at the end of the corresponding play. E.g., in \cref{fig:match-penny-3-die-a} in the play resulting from the path $0, H, T$ the payoff is $0$ because the team loses. It is $1$ when they win. 

Imperfect information is expressed using a dotted line: a player cannot distinguish between two nodes joined by a dotted line.
For e.g., in \cref{fig:match-penny-3-die-a} the dotted red line joining all of Alice's nodes indicates that Alice cannot observe the die outcome. Similarly, the blue dotted line for Bob indicates, he neither observes the outcome of the die, nor the side of the coin chosen by Alice. These sets of indistinguishable nodes are called \emph{information sets}.

\section{Background and notations}
\label{sec:background}


\begin{figure}
\tikzset{
triangle/.style = {regular polygon,regular polygon sides=3,draw,inner sep = 2},
circ/.style = {circle,fill=cyan!10,draw,inner sep = 3},
term/.style = {circle,draw,inner sep = 1.5,fill=black},
sq/.style = {rectangle,fill=gray!20, draw, inner sep = 4}
}

\begin{subfigure}{.45\columnwidth}
\centering
\begin{tikzpicture}[scale=0.9]
\tikzstyle{level 1}=[level distance=9mm,sibling distance = 22mm]
\tikzstyle{level 2}=[level distance=7mm,sibling distance=10mm]
\tikzstyle{level 3}=[level distance=7mm,sibling distance=6mm]
\tikzstyle{level 4}=[level distance=7mm,sibling distance=5mm]


\begin{scope}[->, >=stealth]
\node (0) [circ] {}
child {
  node (00) [triangle] {}
  child {
    node (000) [circ] {}
    child {
      node (0000) [term, label=below:{}] {}
      edge from parent node [left] {\scriptsize $c$}
    }
    child {
      node (0001) [term, label=below:{}] {}
      edge from parent node [right] {\scriptsize $d$}
      }
    edge from parent node [left] {}
  }
  child {
    node (001) [circ] {}
    child {
      node (0010) [term, label=below:{}] {}
      edge from parent node [left] {\scriptsize $c$}
    }
    child {
      node (0011) [term, label=below:{}] {}
      edge from parent node [right] {\scriptsize $d$}
      }
    edge from parent node [right] {} 
  }
  edge from parent node [above] {\scriptsize$a$}
}
child {
  node (01) [triangle] {}
   child {
     node (010) [circ] {}
     child {
      node (0100) [term, label=below:{}] {}
      edge from parent node [left] {\scriptsize $e$}
    }
    child {
      node (0101) [term, label=below:{}] {}
      edge from parent node [right] {\scriptsize $f$}
      }
    edge from parent node [left] {}
  }
  child {
    node (011) [circ] {}
    child {
      node (0110) [term, label=below:{}] {}
      edge from parent node [left] {\scriptsize $e$}
    }
    child {
      node (0111) [term, label=below:{}] {}
      edge from parent node [right] {\scriptsize $f$}
      }
    edge from parent node [right] {} 
  }
  edge from parent node [above] {\scriptsize$b$}
}
;
\end{scope}


  \node[fit=(0),dashed,thick,red, draw, circle,inner sep=1pt] {};
\draw [dashed, thick, blue, in=150,out=30] (000) to (001) ;
\draw [dashed, thick, ForestGreen, in=150,out=30] (010) to (011);

\node [black] at (0,0.35) {\scriptsize $r$};
\node [black] at (-1,-0.55) {\scriptsize $u_1$};
\node [black] at (1, -0.55) {\scriptsize $u_2$};
\node [black] at (-2, -1.5) {\scriptsize $u_3$};
\node [black] at (-.25, -1.5) {\scriptsize $u_4$};

\node [black] at (0.25, -1.5) {\scriptsize $u_5$};
\node [black] at (2, -1.5) {\scriptsize $u_6$};

\node [red] at (0,-.5) {\scriptsize $I_1$};
\node [blue] at (-1.1,-1.6) {\scriptsize $I_2$};
\node [ForestGreen] at (1.1,-1.6) {\scriptsize $I_3$};

\end{tikzpicture}

\caption{$\Max$ with perfect recall}
\label{fig-allexmp-pftrec}
\end{subfigure}
\quad
\begin{subfigure}{.45\columnwidth}
\centering
\begin{tikzpicture}
\tikzstyle{level 1}=[level distance=7mm,sibling distance = 10mm]
\tikzstyle{level 2}=[level distance=7mm,sibling distance=10mm]
\tikzstyle{level 3}=[level distance=7mm,sibling distance=15mm]
\tikzstyle{level 4}=[level distance=7mm,sibling distance=8mm]


\begin{scope}[->, >=stealth]
\node (0) [circ] {}
child{
  node (1) [circ] {}
  child{
    node (3) [term, label=below:{}] {}
    edge from parent node [left] {\scriptsize $a$}
  }
  child{
    node (4) [term,label=below:{}] {}
    edge from parent node [right] {\scriptsize $b$}
  }
  edge from parent node [left] {\scriptsize $a$}
}
child{
  node (2) [term, label=below:{}] {}
  edge from parent node [right] {\scriptsize $b$}
}
;
\end{scope}

\draw [dashed, thick, blue, in=10,out=-100] (0) to (1);

\node [black] at (0,0.25) {\scriptsize $r$};
\node [black] at (-.9,-0.6) {\scriptsize $u_1$};

\node [blue] at (.1,-.6) {\scriptsize $I_1$};

\end{tikzpicture}
\caption{$\Max$ with absentmindedness}
\label{fig-allexmp-absentm}
\end{subfigure}

\caption{Recalls of $\Max$}
\label{fig:recall-examples}
\end{figure}
This section presents the formal definitions. The single team matching-unmatching pennies game has only one player and chance nodes, but in general we will talk about zero-sum two player games. As in \cref{fig:2-p-shuffle}, there are two players $\Max$ (circle nodes) and $\Min$ (square nodes). The payoff at the leaf, is the amount $\Min$ loses and $\Max$ gains. The goal of $\Max$ is to maximize the expected payoff whereas $\Min$ wishes to minimize it. In \cref{fig:match-penny-3-die} $\Max$ was the team consisting of Alice and Bob.

In this paper, we mainly work with \emph{game-structures} and not games
themselves. Game-structures are essentially games sans the numerical
quantities. Any game on a game structure can be represented symbolically as in shown \cref{fig:alossSpan-a} with symbolic payoffs $z_i$s and symbolic chance probabilities $p_i$s (with constraints on $p_i$'s). An extensive
form game can be obtained from a game structure by plugging in values for $z_i$s and $p_i$s.  We work with game structures because the
notions of perfect recall and imperfect recall can be determined
simply by looking at the game-structure.

Formally, a game-structure $\Tt$ is a tuple $(V, L, r, A, E, \Ii)$
where $V$ is a finite set of non-terminal nodes partitioned as
$V_{\Max}$, $V_{\Min}$ and $V_{\chance}$; $L$ is a finite set of leaves;
$r \in V$ is a root node; $A = A_{\Max} \cup A_{\Min}$ is a finite set
of actions; $E \incl V \times (V \cup L)$ is an edge
relation that induces a directed tree; edges originating from $V_{\Max} \cup V_{\Min}$ are labelled with actions from $A$; we write $u \xra{a} v$ if
$(u, v)$ is labelled with $a$, and assume that there is no incoming edge
$u \xra{} r$ to the root node $r$; $\Ii = \Ii_{\Max} \cup \Ii_{\Min}$
is a set of information sets for $i \in \{ \Max, \Min \}$, each
information set $I \in \Ii_i$ is a subset of vertices belonging to
$i$, i.e. $I \incl V_i$, and moreover, the set of information sets
$\Ii_i$ partitions $V_i$. E.g., in \cref{fig-allexmp-pftrec},
$\Ii_{\Max} = \{I_1, I_2, I_3\}$ and $I_1 = \{r\}, I_2 = \{u_3, u_4\}$
and $I_3 = \{u_5, u_6\}$. We can understand these information sets as a signal that the player receives when she reaches a node in it. On receiving the signal, the player knows the actions that are available to play at that position. 

An information set models the fact that a player cannot distinguish
between the nodes within it. Therefore, the set of outgoing actions
from each node in an information set is required to be the same. This
allows us to define $\act(I)$ as the set of actions available at
information set $I$. E.g., in \cref{fig-allexmp-pftrec},
$\act(I_2) = \{c, d\}$. For technical convenience, we make a second
assumption: for all $I, I' \in \Ii$ with $I \neq I'$, we have
$\act(I) \cap \act(I') = \emptyset$. Therefore, the actions identify
the information sets. With this assumption, in \cref{fig:match-penny-3-die}, the actions of Alice should be seen as $H_A, T_A$ and those of Bob's as $H_B, T_B$. But we omit the subscripts in the figure for clarity. 
\begin{definition}[Extensive form games]\label{def:ext-form-games}
  A two-player zero-sum game in extensive form is a tuple
  $(\Tt,\d, \Uu)$ where $\Tt$ is a game-structure, $\d$ is the
  \emph{chance probability} associating to each $\chance$ node, a
  probability distribution on the outgoing actions, and
  $\Uu : L \mapsto \Rat $ is the utility function associating a payoff
  to each leaf.
\end{definition}

The \emph{size} of a game is the sum of the bit-lengths of all chance probabilities and leaf
payoffs in it. A \emph{behavioral strategy} for player $\Max$ ($\Min$ resp.) assigns a probability
distribution to $\act(I)$ for each $I \in \Ii_{\Max}$ ($\Ii_{\Min}$ resp.). Once we fix behavioral strategies $\sigma$ and $\tau$ for $\Max$ and $\Min$ respectively,
each edge in the game has an associated probability of being taken,
given by the corresponding strategy or $\chance$. The probability of reaching a leaf $u \in L$ is given by the product of all the numbers along the path to the leaf. Consider \cref{fig:shuffle-a}.
Let $\sigma$ assign $\frac{1}{4}$ to $b$ and $\frac{3}{4}$ to
$\bar{b}$; $0$ and $1$ to $c$ and $\bar{c}$, and $\frac{1}{3}$ to $a$
and $\frac{2}{3}$ to $\bar{a}$. The probability of reaching the leaf $b \bar{a}$ is 
then: $p_1 \times \frac{1}{4} \times \frac{2}{3}$. For a leaf $u$, we denote this quantity by
$\prob_{\sigma, \tau}(u)$. The \emph{expected payoff} $\Ee(\s, \t)$
when $\Max$ plays $\sigma$ and $\Min$ plays $\tau$, then equals
$\sum_{u \in L} \prob_{\s, \t} (u) \Uu(u)$. The solution concept that we
will consider in this paper is the notion of maxmin.
The \emph{maxmin value} of a game is the
following: \[\max\limits_{\s}\min\limits_{\t}\Ee(\s,\t)\] where
$\s,\t$ are behavioral strategies of $\Max$ and $\Min$ respectively. A
strategy of $\Max$ which provides the maxmin value is called a
\emph{maxmin strategy}. In one-player games, we only have $\Max$ player and the maxmin value of the game is $\max\limits_{\s}\Ee(\s)$. For one-player non-absentminded games, the maxmin value can be in fact obtained by a \emph{pure strategy} -- pure strategies are special cases of behavioural strategies which assign either $0$ or $1$ to each action~\cite{KollerMegiddo::1992}.

The maxmin value of the game in \cref{fig:match-penny-3-die-a} is $\frac{2}{3}$ since Alice and Bob can win at most in 2 of the 3 die rolls by playing matching sides. Another way to see this is to consider the four possible pure strategies $HH, HT, TH, TT$, which induce payoffs $\frac{2}{3}$, $\frac{1}{3}$, $\frac{1}{3}$ and $\frac{2}{3}$ respectively. Now since, in the rest of the following two versions, the team has more information \footnote{This can be observed by the fact that information sets in each version are refinements of the previous versions.} they can guarantee at least $\frac{2}{3}$ by playing the same strategy. Interestingly, one can observe (by enumerating all pure strategies) that they cannot do better than that in any version. 
\paragraph*{Histories and recalls.} We now move on to describing the
various types of imperfect information, based on what the player
remembers about her history. A node $w \in V$ is reached by a unique
path from the root: $r = v_0 \xra{} v_1 \xra{} \cdots \xra{} v_n =
w$. Let $v_{i_1}, v_{i_2}, \dots, v_{i_k}$ be the vertices in this
sequence which do not belong to $\chance$. Then,
$\his(w) = a_{1} a_{2} \cdots a_{k-1}$, where $v_{i_j} \xra{a_j} v_{i_{j + 1}}$.
For a player $i \in \{\Max, \Min\}$ the history of $i$ at $w$, denoted
by $\his_i(w)$, is the sequence of player $i$'s actions in the path to
$w$, which is simply the sub-sequence of $\his(w)$ restricted to
actions from $A_i$. E.g.: in \cref{fig-allexmp-pftrec},
$\his_{\Max}(u_3) = \his_{\Max}(u_4) = a$; in
\cref{fig:shuffle-a}, $\his_{\Max}(u_3) = b$ and $\his_{\Max}(u_2) = \epsilon$, the empty sequence. It is important to remark that this definition uses the assumption that actions determine information sets -- otherwise, we would need to incorporate the information sets that were visited along the way, into the history.

Let $\Hh$ denote the set of all histories and $\Hh_i$ be the set of
all histories of player $i$. For an information set $I \in \Ii_i$ let
$\Hh(I) = \{ \his(u) \mid u \in I\}$ be the set of histories of all
nodes in $I$. Similarly, we can define $\Hh_i(I)$ with respect to
$\Hh_i$. Let $\Hh(L)$ denote the set of all leaf histories.
When $\Hh_i(I)$ has multiple histories, at a node $v \in I$ the player
does not remember which history she traversed to reach $v$. Hence the
player loses information. For two
nodes $u$ and $v$ in $I$, comparing $\his_i(u)$ and $\his_i(v)$
reveals the loss or retention of previously withheld information at
the respective nodes. To capture this there are different notions of
\emph{recall}.

\emph{Perfect recall.} Player $i$ is said to have \emph{perfect
  recall} ($\pfr$) if for every $I \in \Ii_i$, and every pair of
distinct vertices $u, v \in I$, we have $\his_i(u) = \his_i(v)$,
i.e. $|\Hh_i(I)| = 1$.  Otherwise, the player is said to have imperfect
recall.  \cref{fig-allexmp-pftrec} is an example of a perfect recall
game.
\emph{Imperfect recall.} \cref{fig:shuffle-a} gives an example of
a game-structure that has imperfect recall. Notice that states $u_3$
and $u_4$ lie in the same information set $I_3$, but the sequence of
the player's actions leading to these states is different: history at
$u_3$ is $b$, whereas at $u_4$ it is $\bar{b}$. 
Within imperfect recall,
there are distinctions. The imperfect recall in
~\cref{fig:shuffle-b} and the one in ~\cref{fig:shuffle-a}
are in some sense different: in ~\cref{fig:shuffle-b}, the
inability to distinguish between the two nodes in $I_1$ can be traced back to
a point in the past where she forgets her own action from some
information set ($I_3$ in this case), whereas in \cref{fig:shuffle-a}, the player has been able to
distinguish between the two outcomes of the $\chance$ node, but later
forgets at $I_3$ where she started from, leading to four histories $b,\bar{b},c$ and $\bar{c}$ at $I_3$.

\emph{A-loss recall.} Game-structures as in \cref{fig:shuffle-b} are said to have \emph{A-loss
  recall}. A consequence of having A-loss recall is that a player always remembers any new information
gained from $\chance$ outcomes, which is not the case
in~\cref{fig:shuffle-a}. Player $i$ has \emph{A-loss recall}
($\alr$) if for all $I \in \Ii_i$, and every pair of distinct vertices
$u, v \in I$, either $\his_i(u) = \his_i(v)$, or $\his_i(u)$ is of the
form $s a s_1$, and $\his_i(v)$ of the form $s b s_2$, where
$a, b \in \act(I')$ for some $I' \in \Ii$, with $a \neq b$. The game in \cref{fig:match-penny-3-die-a} has A-loss recall, whereas the others, \cref{fig:match-penny-3-die-b} and \cref{fig:match-penny-3-die-c} do not.  

Finally, player $i$ is said to be
\emph{non-absentminded} ($\nam$) if $\forall u, v \in V_i$ with $u$
lying on the path to $v$, the information set that $u$ belongs to is
different from the information set that $v$ belongs
to, i.e. all nodes of $i$ on a path from $r$ to leaf node lie in distinct
information sets. \cref{fig-allexmp-absentm} is an example where
$\Max$ is absentminded, since both $r$ and $u_1$ lie in the same
information set. Notice that $\pfr$ implies $\alr$, which in turn implies implies $\nam$. 

When $\Max$ and $\Min$ have recalls $R_{\Max}, R_{\Min} \in \{ $\pfr$,~$\alr$,~$\nam$ \} $
respectively we will denote the game as a
$(R_{\Max}, R_{\Min})$-game. A one-player game with recall $R$ is denoted as $R$-game. In this paper we are only concerned with one-player $\nam$-games and two-player $(\nam,\nam)$-games. Let us now recall some known
results.
\begin{itemize}\item A maxmin solution in a $(\pfr,\alr)$-game can be computed in
  polynomial- time~\\\cite{KollerMegiddo::1992,vonStengel::1996,kaneko1995behavior}. As a corollary, an optimal solution in a one-player $\alr$-game can be computed in
  polynomial-time ~\cite{kaneko1995behavior}.

\item The maxmin decision problem for $(\nam,\nam)$-games is both
  $\NP$-hard~\cite{KollerMegiddo::1992} and

  $\sqsum$-hard~\cite{GPS20} \footnote{$\sqsum$ is the decision problem of checking
    if the sum of the square roots of $k$ positive integers is less
    than another positive number}.
   The $\NP$-hardness and the $\sqsum$-hardness hold even for
  $(\alr,\pfr)$-games~\cite{Cermak::2018,GPS20}. The maxmin decision problem for one-player $\nam$-games is
  $\NP$-complete~\cite{KollerMegiddo::1992}.
\end{itemize}

Our core idea is to view game structures through the polynomials they generate. 
\paragraph*{Leaf monomials.} In a game structure, assigning variable $x_a$ to each action
$a$, the monomial obtained by taking the product of all $x_a$ along the path to each leaf $t$ is called
a \emph{leaf monomial}, and denoted as $\mu(t)$. E.g., the leaf
monomials of the game-structure in \cref{fig-allexmp-pftrec} are
$\{ x_ax_c, x_ax_d, x_bx_e, x_bx_f\}$. For a game structure $\Tt$, we will write $X(\Tt)$ for the set of leaf monomials. For a game $G$, let
$\prob_{\chance}(t)$ denote the product of $\chance$ probabilities in
the path to $t$. The polynomial given by
$\sum\limits_{t \in L} \prob_{\chance}(t) \cdot \Uu(t) \cdot \mu(t)$
is called the \emph{payoff polynomial} of a game. 
A constraint of the form $\sum\limits_{a \in \act(I)} x_a = 1$ for an
information set $I$ will be called a \emph{strategy constraint}. Any non-negative valuation satisfying these constraints gives a behavioral strategy to the players.
The maxmin value in a game can be given by the maxmin of the payoff polynomial over all possible values satisfying the strategy constraints.

\paragraph*{Overview of our work}
In this work, our mantra for simplifying games is to find simpler games with same payoff polynomials (upto renaming of variables). Leaf monomials are the building blocks of payoff polynomials. We give methods to generate from a given game-structure $\Tt$, a transformed game-structure $\Tt'$ with A-loss recall such that: either $\Tt'$ has the same set of leaf monomials (Section~\ref{sec:shuffled-loss-recall}), or each leaf monomial of $\Tt$ is a linear combination of the leaf monomials of $\Tt'$ (Section~\ref{sec:span}).


\section{Shuffled A-loss recall}
\label{sec:shuffled-loss-recall}


\begin{figure}
\begin{subfigure}{0.5\columnwidth}
\centering
\tikzset{
triangle/.style = {regular polygon,regular polygon sides=3,draw,inner sep = 2},
circ/.style = {circle,fill=cyan!10,draw,inner sep = 3},
term/.style = {circle,draw,inner sep = 1.5,fill=black},
sq/.style = {rectangle,fill=gray!20, draw, inner sep = 4}
}

\begin{tikzpicture}[scale=0.85]
\tikzstyle{level 1}=[level distance=9mm,sibling distance = 22mm]
\tikzstyle{level 2}=[level distance=7mm,sibling distance=10mm]
\tikzstyle{level 3}=[level distance=7mm,sibling distance=6mm]
\tikzstyle{level 4}=[level distance=7mm,sibling distance=5mm]


\begin{scope}[->, >=stealth]
\node (0) [triangle] {}
child {
  node (00) [circ] {}
  child {
    node (000) [circ] {}
    child {
      node (0000) [term, label=below:{\scriptsize $z_1$}] {}
      edge from parent node [left] {\scriptsize $a$}
    }
    child {
      node (0001) [term, label=below:{\scriptsize $z_2$}] {}
      edge from parent node [right] {\scriptsize $\bar{a}$}
      }
    edge from parent node [left] {\scriptsize $b$}
  }
  child {
    node (001) [circ] {}
    child {
      node (0010) [term, label=below:{\scriptsize $z_3$}] {}
      edge from parent node [left] {\scriptsize $a$}
    }
    child {
      node (0011) [term, label=below:{\scriptsize $z_4$}] {}
      edge from parent node [right] {\scriptsize $\bar{a}$}
      }
    edge from parent node [right] {\scriptsize $\bar{b}$} 
  }
  edge from parent node [above] {\scriptsize $p_1$}
}
child {
  node (01) [circ] {}
   child {
     node (010) [circ] {}
     child {
      node (0100) [term, label=below:{\scriptsize $z_5$}] {}
      edge from parent node [left] {\scriptsize $a$}
    }
    child {
      node (0101) [term, label=below:{\scriptsize $z_6$}] {}
      edge from parent node [right] {\scriptsize $\bar{a}$}
      }
    edge from parent node [left] {\scriptsize $c$}
  }
  child {
    node (011) [circ] {}
    child {
      node (0110) [term, label=below:{\scriptsize $z_7$}] {}
      edge from parent node [left] {\scriptsize $a$}
    }
    child {
      node (0111) [term, label=below:{\scriptsize $z_8$}] {}
      edge from parent node [right] {\scriptsize $\bar{a}$}
      }
    edge from parent node [right] {\scriptsize $\bar{c}$} 
  }
  edge from parent node [above] {\scriptsize $p_2$}
}
;
 \node[fit=(00),dashed,thick,blue, draw, circle,inner sep=1pt] {};
  \node[fit=(01),dashed,thick,red, draw, circle,inner sep=1pt] {};
\end{scope}

\draw [dashed, thick, ForestGreen, in=150,out=30] (000) to (001);
\draw [dashed, thick, ForestGreen, in=150,out=30] (001) to (010);
\draw [dashed, thick, ForestGreen, in=150,out=30] (010) to (011);

\node [black] at (0,0.35) {\scriptsize $r$};
\node [black] at (-1.5,-0.55) {\scriptsize $u_1$};
\node [black] at (1.5, -0.55) {\scriptsize $u_2$};
\node [black] at (-2, -1.6) {\scriptsize $u_3$};
\node [black] at (-.25, -1.7) {\scriptsize $u_4$};

\node [black] at (0.25, -1.7) {\scriptsize $u_5$};
\node [black] at (2, -1.6) {\scriptsize $u_6$};

\node [ForestGreen] at (0,-1.1) {\scriptsize $I_3$};
\node [blue] at (-.55,-.9) {\scriptsize $I_1$};
\node [red] at (.55,-.9) {\scriptsize $I_2$};

\end{tikzpicture}
\caption{$\Max$ without $\alr$ but has $\salr$}
\label{fig:shuffle-a}
\end{subfigure}
\begin{subfigure}{0.48\columnwidth}
\centering
\tikzset{
triangle/.style = {regular polygon,regular polygon sides=3,draw,inner sep = 2},
circ/.style = {circle,fill=cyan!10,draw,inner sep = 3},
term/.style = {circle,draw,inner sep = 1.5,fill=black},
sq/.style = {rectangle,fill=gray!20, draw, inner sep = 4}
}

\begin{tikzpicture}[scale=0.85]
\tikzstyle{level 1}=[level distance=9mm,sibling distance = 22mm]
\tikzstyle{level 2}=[level distance=7mm,sibling distance=10mm]
\tikzstyle{level 3}=[level distance=7mm,sibling distance=6mm]
\tikzstyle{level 4}=[level distance=7mm,sibling distance=5mm]


\begin{scope}[->, >=stealth]
\node (0) [circ] {}
child {
  node (00) [triangle] {}
  child {
    node (000) [circ] {}
    child {
      node (0000) [term, label=below:{\scriptsize $z_1$}] {}
      edge from parent node [left] {\scriptsize $b$}
    }
    child {
      node (0001) [term, label=below:{\scriptsize $z_3$}] {}
      edge from parent node [right] {\scriptsize $\bar{b}$}
      }
    edge from parent node [left] {}
  }
  child {
    node (001) [circ] {}
    child {
      node (0010) [term, label=below:{\scriptsize $z_5$}] {}
      edge from parent node [left] {\scriptsize $c$}
    }
    child {
      node (0011) [term, label=below:{\scriptsize $z_7$}] {}
      edge from parent node [right] {\scriptsize $\bar{c}$}
      }
    edge from parent node [right] {} 
  }
  edge from parent node [above] {\scriptsize $a$}
}
child {
  node (01) [triangle] {}
   child {
     node (010) [circ] {}
     child {
      node (0100) [term, label=below:{\scriptsize $z_2$}] {}
      edge from parent node [left] {\scriptsize $b$}
    }
    child {
      node (0101) [term, label=below:{\scriptsize $z_4$}] {}
      edge from parent node [right] {\scriptsize $\bar{b}$}
      }
    edge from parent node [left] {}
  }
  child {
    node (011) [circ] {}
    child {
      node (0110) [term, label=below:{\scriptsize $z_6$}] {}
      edge from parent node [left] {\scriptsize $c$}
    }
    child {
      node (0111) [term, label=below:{\scriptsize $z_8$}] {}
      edge from parent node [right] {\scriptsize $\bar{c}$}
      }
    edge from parent node [right] {} 
  }
  edge from parent node [above] {\scriptsize $\bar{a}$}
}
;
\end{scope}


\node[fit=(0),dashed,thick,ForestGreen, draw, circle,inner sep=1pt] {};
\draw [dashed, thick, blue, in=150,out=30] (000) to (010);
\draw [dashed, thick, red, in=150,out=30] (001) to (011);

%

\node [ForestGreen] at (0,-0.6) {\scriptsize $I_3$};
\node [blue] at (-0.35,-1) {\scriptsize $I_1$};
\node [red] at (0.35,-1) {\scriptsize $I_2$};

\node[black] at (-1.5,-.95) {\scriptsize $p_1$};
\node[black] at (-.73,-.95) {\scriptsize $p_2$};

\node[black] at (1.5,-.95) {\scriptsize $p_2$};
\node[black] at (.73,-.95) {\scriptsize $p_1$};
\end{tikzpicture}
\caption{$\Max$ with $\alr$}
\label{fig:shuffle-b}
\end{subfigure}
\caption{Equivalent $\alr$ game using $\salr$ for game without $\alr$ }
\label{fig:shuffle}
\end{figure}
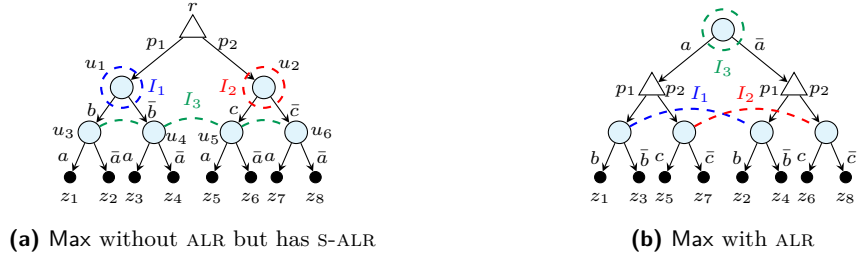

We start with an example. The game-structure in \cref{fig:shuffle-a}
is an equivalent game of version \textbf{II} of the matching-unmatching
game (\cref{fig:match-penny-3-die-b}) obtained by merging die outcome $0$ and $1$ (and renaming $H, T$), with $p_1 = \frac{2}{3}, p_2 = \frac{1}{3}$. The game
does not have $\alr$: we have
$\Hh_{\Max}(I_3) = \{b, \bar{b}, c, \bar{c}\}$. Since
$\{b, \bar{b}\}$ and $\{c, \bar{c}\}$ are from different information
sets, the pair of histories $b$ and $c$, for instance, is a witness
for no $\alr$. The player forgets what she knew about $\chance$ actions. Now,
consider the game-structure in~\cref{fig:shuffle-b}, obtained by
\emph{shuffling} the actions ($a$ goes above $b$ and $c$). This
game-structure has $\alr$. The crucial observation is that both
the game-structures, \cref{fig:shuffle-a} and \cref{fig:shuffle-b},
lead to the same leaf monomials\footnote{Two monomials are same if their sets of variables are same}:
$\{x_ax_b,~x_ax_{\bar{b}},~x_{\bar{a}}x_b,~x_{\bar{a}}x_{\bar{b}},~x_ax_c,~
x_ax_{\bar{c}},~x_{\bar{a}}x_c,~x_{\bar{a}}x_{\bar{c}} \}$. Similarly,
in \cref{fig:match-penny-3-die-b}, by shuffling the turns of Alice and
Bob, we get an $\alr$ recall game that induces the same leaf
monomials. 

We say that the game-structure of ~\cref{fig:shuffle-a} has
\emph{shuffled A-loss recall}. Even though the game-structure
originally does not have $\alr$, it can be shuffled in some way
to get an $\alr$ structure.  Not every game-structure has shuffled
A-loss recall. In this section, we provide a polynomial-time algorithm
to identify whether a game-structure has shuffled A-loss recall. If
the answer is yes, the algorithm also computes the shuffled
game-structure. As a result, we are able to show that one-player
shuffled A-loss recall games can be solved in
polynomial-time. 
We will keep our discussion to one-player games played by $\Max$, and
later in \cref{sec:two-player} discuss extensions to two-player
games. 
We will require few notions and notations, which we introduce
gradually as we need them.

For a game structure $\Tt$, we write $L_{\Tt}$ for the set of its
leaves. Define $|\Tt|$, the \emph{size} of a game structure $\Tt$, to be
$|L_{\Tt}|$, the total number of its leaves. We work with history sequences originating from game structures.
Fix a finite set of information
sets $\Ii$ and a set of actions $\act(I)$ for each $I \in \Ii$.~Recall that the action sets of distinct information sets are disjoint. 
Let $A = \biguplus_{I \in \Ii} \act(I)$. A \emph{sequence} is a finite word over
$A^*$ that contains at most one letter from each $\act(I)$ \footnote{
 We restrict to such words in $A^*$ since histories in an $\nam$ game structure have this property}.

Let $s[i]$ denote the $i$th action in $s$. For sequences $s_1$ and $s_2$ of length $k$, we say
$s_2$ is a permutation of $s_1$ if $\exists$~a bijective function
$\rho : \{1,\dots,k\} \mapsto \{1,\dots,k\}$ such that $\forall i$,
$s_1[i] = s_2[\rho(i)]$.

\begin{definition}[Shuffled A-loss recall]\label{def:salr}
  A game structure $\Tt$ is said to have shuffled A-loss
  recall~(\salr) if $\exists$ a game structure $\Tt'$ with
  $|\Tt| = |\Tt'|$ such that
  \begin{itemize}   \item $\Tt'$ has A-loss recall
  \item There is a bijection $f : L_{\Tt} \mapsto L_{\Tt'}$ such that
    $\forall t \in L_{\Tt}$, $\his(f(t))$ is a permutation of
    $\his(t)$.
  \end{itemize}
\end{definition}
When $\Tt$ has $\salr$ we call the structure $\Tt'$ an $\salr$ witness of $\Tt$. The following lemma is a consequence on leaf monomials.

\begin{restatable}{lemma}{shuffleSameLeafMonomials}\label{lem:alr-salr-same-leaf-monomials}
  Suppose $\Tt$ has $\salr$ with $\salr$ witness $\Tt'$, then $\Tt$ and $\Tt'$ have the same set of leaf
  monomials.
\end{restatable}

\paragraph*{Detecting S-ALR.}

Checking whether a structure $\Tt$ has $\pfr$ or $\alr$ can be done in
polynomial-time, simply by checking histories at every information
set. On the other hand, from the definition of $\salr$, it is not
immediate if one could test it efficiently. One approach could be
finding good permutations for each leaf history in order to get the $\salr$ witness but this could potentially lead to
exponentially many checks.  In the following discussion we will
provide a polynomial-time algorithm to test $\salr$ in a structure.

\begin{theorem}\label{thm:shuffle-detection-ptime} Given a game structure
  $\Tt$, there is an algorithm that checks if $\Tt$ has $\salr$ in
  time $O(|\Tt|)$. Moreover if $\Tt$ does have $\salr$, this algorithm
  also outputs an $\salr$ witness structure $\Tt'$.
\end{theorem}

To prove Theorem~\ref{thm:shuffle-detection-ptime}, we will work with
the history sequences. For a game structure $\Tt$, recall that
$\Hh(L_{\Tt})$ is the set of all leaf histories in $\Tt$. We will
construct a set of leaf histories $H'$ such that $H' = \Hh(L_{\Tt'})$
for some $\alr$ structure $\Tt'$. Since we work with sequences, and
not game-structures themselves, we will need a notion of $\alr$ for
sets of words. When we are given a game, $\alr$ can be detected by
looking at the histories. When we are given a set of histories, it is
not as direct. We will need to determine some structure inside the
sequences.

\paragraph*{ALR on sequence sets.} 

Two sequences $s_1, s_2$ over $A$ are said to be \emph{connected} if
there is some information set $I$ such that both $s_1$ and $s_2$
contain an action from $\act(I)$. E.g., suppose
$\act(I_1) = \{a, b\}$, $\act(I_2) = \{c, d\}$, $\act(I_3) = \{e,f\}$; then $s_1 = ac$ and
$s_2 = eb$ are connected since $s_1$ contains $a$, $s_2$ contains $b$,
both of which are in $\act(I_1)$. Sequences $s_1, s_2$ are said to be
\emph{disconnected} if they are not connected. For the same alphabet
as before, let $s_3 = e$, then $s_1$ and $s_3$ are disconnected.

We say that a set of sequences $S$ is connected if for every disjoint
partition of $S$ as $S_1 \uplus S_2$ (where $S_1, S_2$ are non-empty),
there exist $s_1 \in S_1$ and $s_2 \in S_2$ such that $s_1$ and $s_2$
are connected. E.g., consider $S = \{ac, eb, e\}$ with
information sets as above. This set $S$ is connected, even though $ac$
and $e$ are not connected. For a set $S$, we can construct an
undirected graph as follows: each $s \in S$ is a vertex, and there is
an edge between $s_1, s_2 \in S$ if they are connected. Notice that a
set $S$ is connected iff there is a path between any two vertices in
this graph. This interpretation allows to decompose $S$ uniquely as
$S = \biguplus_i S_i$ where each $S_i$ is a
\emph{maximal connected component} in the graph.

We can now give a definition of $\alr$ on sequences. This is defined
inductively as follows. The set $\{\epsilon\}$ has $\alr$. A disconnected set
$S$ with decomposition $\biguplus_i S_i$ has $\alr$ if each of
its connected components $S_i$ has $\alr$. A connected set $S$ has
$\alr$ if there exists an $I \in \Ii$ s.t.:
\begin{enumerate} \item every sequence in $S$ starts with $\act(I)$: i.e. each
  $w \in S$ is of the form $a u$ for some $a \in \act(I)$, and
\item the set of continuations of each $a \in \act(I)$ has $\alr$: for
  each $a \in \act(I)$, the set $\{ u \mid au \in S\}$ has $\alr$.
\end{enumerate}
Let us illustrate this definition on examples from
\cref{fig:recall-examples}. In \ref{fig-allexmp-pftrec}, we have the
leaf histories $H_1 = \{ac, ad, be, bf\}$. Notice that $H_1$ is a
connected set. There is an information set $I_1$ with
$\act(I_1) = \{a, b\}$ such that the first condition above is
true. For the second condition, let us look at the continuations:
$H_1^a = \{c, d\}$ and $H_1^b = \{e, f\}$. Both $H_1^a$ and $H_1^b$
are connected and satisfy the first condition. The second condition is
vacuously true for $H_1^a$ and $H_1^b$. This shows $H_1$ has $\alr$
(as expected, since on game-structures, $\pfr$ is a subclass of
$\alr$). Now, let us look at \cref{fig:shuffle-a}. The leaf
histories are given by $H_2 = \{ba,b\bar{a},\bar{b}a,\bar{b}\bar{a},ca,c\bar{a},\bar{c}a,\bar{c}\bar{a}\}$. Observe that
$H_2$ is connected. However, the first condition in the $\alr$
definition does not hold. So $H_2$ is not $\alr$. We can show that the recursive
definition of $\alr$-sets and that of $\alr$ game structures are equivalent.

\begin{restatable}{proposition}{alrSetStructEquiv}\label{prop:alr-set-struct-equiv}
  A game structure $\Tt$ has $\alr$ iff $\Hh(L_{\Tt})$ has $\alr$.
\end{restatable}

\paragraph*{S-ALR on sequence sets.} We can also extend the
definition of $\salr$ to sequence sets.  $S$ has $\salr$ if $\exists$
another set $S'$ such that (i) $S'$ has $\alr$ and (ii) there is a
bijection $f : S \mapsto S'$ where $\forall s \in S,$ $f(s)$ is a
permutation of $s$. We call the set $S'$ an $\salr$ witness of $S$.

Exploiting the recursive definition of $\alr$ sets we will provide
recursive necessary and sufficient conditions for sequence sets to
have $\salr$.  Firstly, one can check for $\salr$ for a set $S$ by
checking $\salr$ for each individual maximal connected components.

\begin{restatable}{proposition}{salrDisc}\label{lem:aloss-shuffle-disconnected}
  Let $S$ be a disconnected set and $S = \biguplus_i S_i$ be the
  decomposition of $S$ into maximal connected components. Then $S$ has
  $\salr$ iff $\forall i,S_i$ has $\salr$ .
\end{restatable}

\begin{restatable}{corollary}{salrDiscApply}\label{cor:alosshuffle-disconnected-apply}
  Let $S =\uplus_i S_i$ be the decomposition of $S$ where each $S_i$
  has $\salr$ witnessed by sets $S'_i$. Then $S$ has $\salr$ and
  $S' = \bigcup_i S'_i$ is an $\salr$ witness of $S$.
\end{restatable}

Secondly, for connected sets we will provide another recursive
condition to have $\salr$. We use another notation. For an action
$a \in A$, we write $S_a := \{u_1 u_2 \mid u_1 a u_2 \in S\}$, for the
set of sequences obtained by picking a sequence in $S$ that contains $a$,
and dropping the $a$ from it.

\begin{restatable}{proposition}{commonact}\label{prop:commonact}
  Let $S$ be a connected set. $S$ has $\salr$ iff there exists an
  information set $I \in \Ii$ such that
  \begin{enumerate}   \item every sequence $s \in S$ contains an action from $\act(I)$,
  \item and for all $a \in \act(I)$, the set $S_a$ has $\salr$.
  \end{enumerate}
\end{restatable}

The above proposition can be naturally translated into an
algorithm. However, this does not yet ensure a polynomial-time
complexity. If there are two information sets $I_1$ and $I_2$ that
satisfy Condition 1 of Proposition~\ref{prop:commonact}, the order in
which we pick them might, in principle, create a difference and hence
one has to guess the right order. The next lemma says this does not
happen.

\begin{restatable}{lemma}{salrPtime}\label{lem:salr-common-act}
  Let $S$ be a connected set and let $I$ be an arbitrary information
  set such that every sequence $s \in S$ contains an action from
  $\act(I)$. Then: $S$ has $\salr$ iff for all $a \in \act(I)$, the
  set $S_a$ has $\salr$. Moreover, if for each $a$, $S'_a$ is $\salr$ witness of $S_a$, then $\bigcup_{a} aS_a'$ is an $\salr$ witness of $S$.  
\end{restatable}


\begin{algorithm}{}  
\caption{Compute shuffled A-loss recall}
\label{algo:alosshuffle}
\begin{algorithmic}[1]
  \STATE \textbf{Input} : $S$
  \STATE \textbf{Output}: $\salr$ witness $S'$ of $S$ if it exists
\IF{$S$ is connected} 
\IF{$\exists I$ such that every $s \in S$ contains an action from
  $\act(I)$ \label{algoShuffleLine:connected}} 
\FOR {$a \in \act(I)$}
\STATE $S_a' \gets $ $\salr$ witness of $S_a$
\ENDFOR
\RETURN $\bigcup_{a} aS_a'$ 
\ELSE \STATE \textbf{EXIT} and report $S$ does not have $\salr$ \label{algoShuffleLine:noshuffle}
\ENDIF
\ELSE
\STATE $S = \uplus S_i$ where each $S_i$ is connected 
\STATE $S'_i \gets $ $\salr$ witness of $S_i$ \label{algoShuffleLine:disconnected}
\RETURN $\cup_i S'_i$
\ENDIF
\end{algorithmic} 
\end{algorithm}


Proposition~\ref{prop:commonact} and Lemma~\ref{lem:salr-common-act} lead
to
Algorithm~\ref{algo:alosshuffle}.      This algorithm runs in time $\Oo(|S|)$, proving
Theorem~\ref{thm:shuffle-detection-ptime}.
Here is an example run of the algorithm on \cref{fig:shuffle-a}. 

We
have $S = \{ ba, b\bar{a}, \bar{b}a, \bar{b}\bar{a}, ca, c\bar{a}, \bar{c}
a, \bar{c} \bar{a} \}$. All the sequences contain an action from
$\act(I_1) = \{a, \bar{a}\}$. For the recursive call, we the set
$H = \{b, \bar{b}, c, \bar{c}\}$. This set is disconnected, with
components $H_1 = \{b, \bar{b}\}$ and $H_2 = \{c, \bar{c}\}$. For
$H_1, H_2$ the same sets are witness for $\salr$. For $H$, the witness
is the same set again. For $S$, the witness is
$aH \biguplus \bar{a}H = \{ab, a\bar{b}, \dots, \bar{a} c, \bar{a}
\bar{c} \}$. This set can be translated to the game
\cref{fig:shuffle-b}, with $\alr$.   Now, back to our
running example \cref{fig:match-penny-3-die-b}. Let $H_{01}, T_{01}$ and
$H_2, T_2$ be Alice's actions out of the information sets obtained
after die roll $0$ or $1$, and $2$ respectively. The induced
sequence set is $\{ H_{01} H, H_{01} T, T_{01} H, T_{01}T, H_2 H, T_2
T \}$, which can be seen to have $\salr$. In fact, for all values for
$n$, the game \textbf{II} has $\salr$.

Theorem~\ref{thm:shuffle-detection-ptime} and the fact that $\alr$
games can be solved in polynomial-time give us the following theorem,
and hence a new polynomial-time solvable class with imperfect
recall.

\begin{restatable}{theorem}{onepShufflePtime}\label{thm:1p-shuffle-ptime}
  The maxmin value in one-player games with $\salr$ can be computed in
  polynomial time.
\end{restatable}


\section{Span}
\label{sec:span}


\begin{figure}
\centering

\begin{subfigure}{.3\columnwidth}
\tikzset{
triangle/.style = {regular polygon,regular polygon sides=3,draw,inner sep = 2},
circ/.style = {circle,fill=cyan!10,draw,inner sep = 3},
term/.style = {circle,draw,inner sep = 1.5,fill=black},
sq/.style = {rectangle,fill=gray!20, draw, inner sep = 4}
}

\begin{tikzpicture}[scale=0.85]
\tikzstyle{level 1}=[level distance=9mm,sibling distance = 22mm]
\tikzstyle{level 2}=[level distance=7mm,sibling distance=10mm]
\tikzstyle{level 3}=[level distance=7mm,sibling distance=6mm]
\tikzstyle{level 4}=[level distance=7mm,sibling distance=5mm]


\begin{scope}[->, >=stealth]
\node (0) [triangle] {}
child {
  node (00) [circ] {}
  child {
    node (000) [circ] {}
    child {
      node (0000) [term, label=below:{\scriptsize $z_1$}] {}
      edge from parent node [left] {\scriptsize $c$}
    }
    child {
      node (0001) [term, label=below:{\scriptsize $z_2$}] {}
      edge from parent node [right] {\scriptsize $\bar{c}$}
      }
    edge from parent node [left] {\scriptsize $a$}
  }
  child {
    node (001) [circ] {}
    child {
      node (0010) [term, label=below:{\scriptsize $z_3$}] {}
      edge from parent node [left] {\scriptsize $d$}
    }
    child {
      node (0011) [term, label=below:{\scriptsize $z_4$}] {}
      edge from parent node [right] {\scriptsize $\bar{d}$}
      }
    edge from parent node [right] {\scriptsize $\bar{a}$} 
  }
  edge from parent node [above] {\scriptsize $p_1$}
}
child {
  node (01) [circ] {}
   child {
     node (010) [circ] {}
     child {
      node (0100) [term, label=below:{\scriptsize $z_5$}] {}
      edge from parent node [left] {\scriptsize $c$}
    }
    child {
      node (0101) [term, label=below:{\scriptsize $z_6$}] {}
      edge from parent node [right] {\scriptsize $\bar{c}$}
      }
    edge from parent node [left] {\scriptsize $b$}
  }
  child {
    node (011) [circ] {}
    child {
      node (0110) [term, label=below:{\scriptsize $z_7$}] {}
      edge from parent node [left] {\scriptsize $d$}
    }
    child {
      node (0111) [term, label=below:{\scriptsize $z_8$}] {}
      edge from parent node [right] {\scriptsize $\bar{d}$}
      }
    edge from parent node [right] {\scriptsize $\bar{b}$} 
  }
  edge from parent node [above] {\scriptsize $p_2$}
}
;
 \node[fit=(00),dashed,thick,blue, draw, circle,inner sep=1pt] {};
  \node[fit=(01),dashed,thick,red, draw, circle,inner sep=1pt] {};
\end{scope}

\draw [dashed, thick, ForestGreen, in=150,out=30] (000) to (010);
\draw [dashed, thick, brown, in=150,out=30] (001) to (011);
\node [black] at (0,0.35) {\scriptsize $r$};
\node [black] at (-1.5,-0.55) {\scriptsize $u_1$};
\node [black] at (1.5, -0.55) {\scriptsize $u_2$};
\node [black] at (-2, -1.6) {\scriptsize $u_3$};
\node [black] at (-.25, -1.7) {\scriptsize $u_4$};

\node [black] at (0.25, -1.7) {\scriptsize $u_5$};
\node [black] at (2, -1.6) {\scriptsize $u_6$};


\node [blue] at (-1.7,-.9) {\scriptsize $I_1$};
\node [red] at (1.7,-.9) {\scriptsize $I_2$};
\node [ForestGreen] at (-0.2,-1) {\scriptsize $I_3$};
\node [brown] at (0.3,-1) {\scriptsize $I_4$};

\end{tikzpicture}
\caption{$\Max$ without $\salr$}
\label{fig:alossSpan-a}
\end{subfigure}

\begin{subfigure}{.6\columnwidth}
\tikzset{
triangle/.style = {regular polygon,regular polygon sides=3,draw,inner sep = 2},
circ/.style = {circle,fill=cyan!10,draw,inner sep = 3},
term/.style = {circle,draw,inner sep = 1.5,fill=black},
sq/.style = {rectangle,fill=gray!20, draw, inner sep = 4}
}

\begin{tikzpicture}[scale=0.8]
\tikzstyle{level 1}=[level distance=9mm,sibling distance = 50mm]
\tikzstyle{level 2}=[level distance=5mm,sibling distance=25mm]
\tikzstyle{level 3}=[level distance=9mm,sibling distance=12mm]
\tikzstyle{level 4}=[level distance=10mm,sibling distance=6mm]


\begin{scope}[->, >=stealth]
\node (0) [circ] {}
child {
  node (00) [circ] {}
  child {
  node (000) [triangle] {}
   child {
     node (0000) [circ] {}
     child {
      node (00000) [term, label=below:{\scriptsize $w_1$}] {}
      edge from parent node [left] {\scriptsize $a$}
    }
    child {
      node (00001) [term, label=below:{\scriptsize $w_2$}] {}
      edge from parent node [right] {\scriptsize $\bar{a}$}
      }
    edge from parent node [left,pos=0.2] {\scriptsize $\frac{1}{2}$}
  }
  child {
    node (0001) [circ] {}
    child {
      node (00010) [term, label=below:{\scriptsize $w_3$}] {}
      edge from parent node [left] {\scriptsize $b$}
    }
    child {
      node (00011) [term, label=below:{\scriptsize $w_4$}] {}
      edge from parent node [right] {\scriptsize $\bar{b}$}
      }
    edge from parent node [right,pos=0.2] {\scriptsize $\frac{1}{2}$} 
     }
  edge from parent node [above] {\scriptsize $d$}
  }
  child {
  node (001) [triangle] {}
   child {
     node (0010) [circ] {}
     child {
      node (00100) [term, label=below:{\scriptsize $w_5$}] {}
      edge from parent node [left] {\scriptsize $a$}
    }
    child {
      node (00101) [term, label=below:{\scriptsize $w_6$}] {}
      edge from parent node [right] {\scriptsize $\bar{a}$}
      }
    edge from parent node [left,pos=0.2] {\scriptsize $\frac{1}{2}$}
  }
  child {
    node (0011) [circ] {}
    child {
      node (00110) [term, label=below:{\scriptsize $w_7$}] {}
      edge from parent node [left] {\scriptsize $b$}
    }
    child {
      node (00111) [term, label=below:{\scriptsize $w_8$}] {}
      edge from parent node [right] {\scriptsize $\bar{b}$}
      }
    edge from parent node [right,pos=0.2] {\scriptsize $\frac{1}{2}$} 
     }
  edge from parent node [above] {\scriptsize $\bar{d}$}
  }
  edge from parent node [above] {\scriptsize $c$}
  }
child {
  node (01) [circ] {}
  child {
  node (010) [triangle] {}
   child {
     node (0100) [circ] {}
     child {
      node (01000) [term, label=below:{\scriptsize $w_9$}] {}
      edge from parent node [left] {\scriptsize $a$}
    }
    child {
      node (01001) [term, label=below:{\scriptsize $w_{10}$}] {}
      edge from parent node [right] {\scriptsize $\bar{a}$}
      }
    edge from parent node [left,pos=0.2] {\scriptsize $\frac{1}{2}$}
  }
  child {
    node (0101) [circ] {}
    child {
      node (01010) [term, label=below:{\scriptsize $w_{11}$}] {}
      edge from parent node [left] {\scriptsize $b$}
    }
    child {
      node (01011) [term, label=below:{\scriptsize $w_{12}$}] {}
      edge from parent node [right] {\scriptsize $\bar{b}$}
      }
    edge from parent node [right,pos=0.2] {\scriptsize $\frac{1}{2}$} 
  }
  edge from parent node [above] {\scriptsize $d$}
}
  child {
  node (011) [triangle] {}
   child {
     node (0110) [circ] {}
     child {
      node (01100) [term, label=below:{\scriptsize $w_{13}$}] {}
      edge from parent node [left] {\scriptsize $a$}
    }
    child {
      node (01101) [term, label=below:{\scriptsize $w_{14}$}] {}
      edge from parent node [right] {\scriptsize $\bar{a}$}
      }
    edge from parent node [left,pos=0.2] {\scriptsize $\frac{1}{2}$}
  }
  child {
    node (0111) [circ] {}
    child {
      node (01110) [term, label=below:{\scriptsize $w_{15}$}] {}
      edge from parent node [left] {\scriptsize $b$}
    }
    child {
      node (01111) [term, label=below:{\scriptsize $w_{16}$}] {}
      edge from parent node [right] {\scriptsize $\bar{b}$}
      }
    edge from parent node [right,pos=0.2,pos=0.2] {\scriptsize $\frac{1}{2}$} 
     }
  edge from parent node [above] {\scriptsize $\bar{d}$}
  }
  edge from parent node [above] {\scriptsize $\bar{c}$}
}
;
\end{scope}


\node[fit=(0),dashed,thick,ForestGreen, draw, circle,inner sep=1pt] {};
\draw [dashed, thick, brown, in=165,out=15] (00) to (01);
\draw [dashed, thick, blue, in=150,out=30] (0000) to (0010);
\draw [dashed, thick, blue, in=150,out=30] (0010) to (0100);
\draw [dashed, thick, blue, in=150,out=30] (0100) to (0110);
\draw [dashed, thick, red, in=150,out=30] (0001) to (0011);
\draw [dashed, thick, red, in=150,out=30] (0101) to (0111);
\draw [dashed, thick, red, in=150,out=30] (0011) to (0101);

%

\node [ForestGreen] at (0.55,0.1) {\scriptsize $I_3$};
\node [brown] at (0,-.8) {\scriptsize $I_4$};
\node [blue] at (-2.8,-1.6) {\scriptsize $I_1$};
\node [red] at (2.8,-1.6) {\scriptsize $I_2$};


\end{tikzpicture}
\caption{$\Max$ with $\alr$}
\label{fig:alossSpan-b}
\end{subfigure}
\caption{Equivalent $\alr$ game using $\alr$-span for game without $\salr$}
\label{fig:span}
\end{figure}
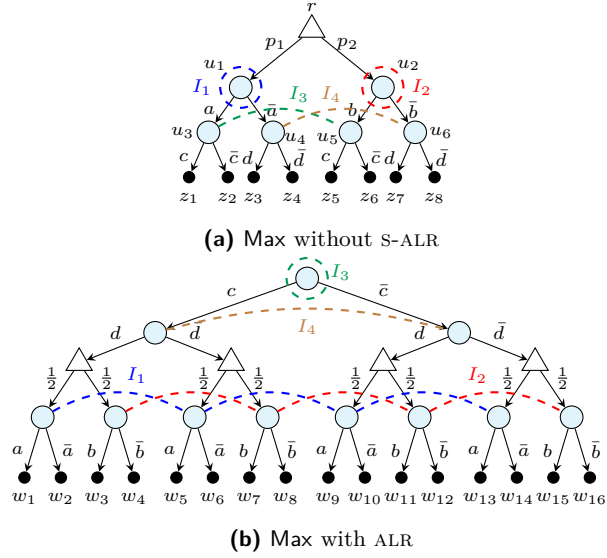

We move on to another way of simplifying game-structures by generalizing $\salr$. The
game-structure in \cref{fig:alossSpan-a} (call it $\Tt_1$) is an equivalent version of game \textbf{III} (\cref{fig:match-penny-3-die-c}). It neither has
$\pfr$, nor $\alr$. Using
Algorithm~\ref{algo:alosshuffle}, we can show that it does not have $\salr$ either. Now, consider the game-structure in
\cref{fig:alossSpan-b} (call it
$\Tt'_1$). It has $\alr$. Each leaf monomial of $\Tt_1$ can be written as 
a linear combination of the leaf monomials of $\Tt'_1$: e.g.,
the leaf monomial $x_ax_{\bar{c}}$ of $\Tt_1$ is equal to
$x_{\bar{c}}x_d x_a + x_{\bar{c}} x_{\bar{d}} x_a$, a combination of leaf
monomials of $\Tt'_1$.  The game-structure $\Tt_1$ is said to be
\emph{spanned by} $\Tt'_1$. This property allows to solve games
derived from the structure $\Tt_1$ by converting them into a game on
$\Tt'_1$ with a suitably designed utility function so that both games
induce the same payoff polynomial, and then solving the resulting
A-loss recall game. This is 
illustrated in \cref{fig:alossSpan-a} and \cref{fig:alossSpan-b}. Results in this section:
\begin{itemize}\item We show that every \nam~ imperfect recall game structure is
  spanned by an A-loss recall structure
  (\cref{thm:existence-alr-span}).  The caveat is that the
  smallest A-loss recall span may be of exponential size: we exhibit a
  family of game structures where this happens
  (\cref{thm:lower-bound}).

\item We provide an algorithm to compute an A-loss recall span of
  smallest size. We show that the associated decision problem is in $\NP$ (\cref{thm:span-NP}). We also identify classes of games with `small' $\alr$-span using a new parameter (\cref{cor:effic-solv-class}). 
\end{itemize}

We will now formally present $\alr$-span.
Similar to last section, we will keep our discussion to one-player
games and later discuss extensions to two players in \cref{sec:two-player}.

Since we will deal with polynomials formed using leaf monomials, we
need the notion of reducing one polynomial to
another.   Recall that since these variables are denoting
behavioral strategies, every valuation to the variables satisfies the
strategy constraints.  We say polynomial $f_1$ reduces to $f_2$ under
strategy constraints if we can get $f_2$ by applying finitely many
substitutions in $f_1$ of the form $\sum\limits_{a \in Act(I)}
x_{a} = 1$. E.g., the polynomial
$x_a x_d x_c ~+~ x_a x_d x_{\bar{c}} ~+~x_{\bar{a}}x_d x_c ~+~ x_{\bar{a}}
x_d x_{\bar{c}}~$ reduces to $x_d$ by applying the substitutions
$x_a + x_{\bar{a}} =1$ and $x_c + x_{\bar{c}} =1$. Observe that, when
$f_1$ reduces $f_2$, they are essentially the same polynomials over
the space with strategy constraints, i.e. they evaluate to the same
value under every assignment of values satisfying the strategy
constraints.   \begin{definition}[$\alr$-span]
  Let $\Tt$ be a game structure with set of leaf monomials
  $X(\Tt)$. We call a structure $\Tt'$ an $\alr$-span of $\Tt$ if
  \begin{itemize}  \item $\Tt'$ has $\alr$.
  \item each monomial in $X(\Tt)$ can be generated by monomials in
    $X(\Tt')$ by linear combinations i.e.  $\forall \mu \in X(\Tt)$
    there exist coefficients
    $ \{c^{\mu}_{\mu'}\}_{\mu' \in X(\Tt')} \in \Real^{|X(\Tt')|}$ such that the
    polynomial $\sum\limits_{\mu' \in X(\Tt')}c^{\mu}_{\mu'}\mu'$ reduces to
    $\mu$ under strategy constraints.
  \end{itemize}

\end{definition}
The game structure $\Tt'_1$ in \cref{fig:alossSpan-b} is an
$\alr$-span of $\Tt_1$ in \cref{fig:alossSpan-a}. E.g., for the
monomial $x_{\bar{b}}x_{\bar{d}}$ in $\Tt_1$, the linear combination
$x_cx_{\bar{d}} x_{\bar{b}} + x_{\bar{c}} x_{\bar{d}} x_{\bar{b}}$
reduces to it by substituting $x_d + x_{\bar{d}} = 1$. In fact,
observe that for any monomial $\mu$ in $X(\Tt_1)$, the sum of the two
monomials in $X(\Tt'_1)$ that contain the actions in $\mu$, reduces to
$\mu$.

The next proposition states that we can use $\alr$-span $\Tt_1'$ to solve
games on structure $\Tt_1$ (see \cref{fig:alossSpan-a} and
\cref{fig:alossSpan-b}). We can assign suitable payoffs $w_i$s in
terms of $z_i$s and $p_i$s
to get an equivalent game. Let $t_i$ be the
leaf that has payoff $w_i$.The payoff $w_i$ is the sum of the quantities
$\frac{\prob_{\chance}(t)c^{\mu(t)}_{\mu(t_i)}\Uu(t)}{\prob_{\chance}(t_i)}$ for all $t$ in
$\Tt_1$, in which $\mu(t_i)$ contributed to generate $\mu(t)$. E.g., $t_1$, with polynomial $x_c x_d x_a$ contributes in generating
only $x_ax_c$ (for the leaf corresponding to $z_1$). Hence $w_1 =
2p_1z_1$. Similarly, $w_5 = 2p_1z_1$. In the payoff polynomial of the second game, since $x_cx_dx_a + x_c\bar{x_d}x_a = x_ax_c$, we will have the term $p_1z_1 x_ax_c$ in the payoff polynomial. This way, we will get the same payoff polynomial as original game, thus reducing the original game to a game on $\Tt'_1$. 

\begin{restatable}{proposition}{spanToGame}\label{prop:span-to-game}
  Let $\Tt$ be a game structure and $\Tt'$ an $\alr$-span of $\Tt$.
  Then for every game $G = (\Tt, \d, \Uu)$ on $\Tt$, there exists a game
  $G' = (\Tt', \d', \Uu')$ on $\Tt'$ such that solving $G$ can be
  reduced to solving $G'$. \end{restatable}

\paragraph*{Existence of ALR-spans.} From the definition, one can
see that when $\Tt$ has $\salr$, an $\salr$ witness is
also an $\alr$-span of $\Tt$. Surprisingly, we can show that every
$\nam$ structure has $\alr$-span.

\begin{restatable}{theorem}{alrspan}\label{thm:existence-alr-span}
  Every $\nam$-game structure $\Tt$ has an $\alr$-span.
\end{restatable}
\begin{proof}[Proof Sketch]
  We can explicitly provide an $\alr$-span $\Tt'$ of $\Tt$.   Let $\Ii$ be the set of information sets in $\Tt$. The structure
  $\Tt'$ has $|\Ii|$ levels of player nodes corresponding to each $I
  \in \Ii$. All nodes in a level are placed in
  one information set. Therefore the leaves of $\Tt'$ are all possible
  monomials using the information sets $\Ii$. Moreover, $\Tt'$ has $\alr$. To generate a monomial $\mu \in X(\Tt)$
  one can combine the monomials of all paths containing actions from
  $\mu$ in $X(\Tt')$.
\end{proof}

\paragraph*{Minimal ALR-span} In order to take advantage of
Proposition~\ref{prop:span-to-game}, one would need to find small
$\alr$-spans. We observe that the $\alr$-span obtained in the proof of
Theorem~\ref{thm:existence-alr-span} has exponential size.  We will
now delve into finding $\alr$-spans of smallest size : a \emph{minimal
  $\alr$-span}.  First we will list some key observations concerning
minimal $\alr$-spans which will lead to an algorithm for computing
one. Then we will show that the exponential blowup in size of
$\alr$-spans is unavoidable in general by exhibiting a class of games
with minimal $\alr$-span of exponential size. Since $\alr$-games are solvable in polynomial time, this aligns with the fact that the maxmin problem for general one-player
  $\nam$-games is $\NP$-hard.

Similar to Section~\ref{sec:shuffled-loss-recall}, we will work
directly with sequence sets. The notions related to span are extended
to sequence sets in a natural manner, i.e. when $\Tt'$ is an
$\alr$-span of $\Tt$, $\Hh(L_{\Tt'})$ is an $\alr$-span of
$\Hh(L_{\Tt})$  

\begin{restatable}{proposition}{minSpanDisc}\label{prop:min-span-disconnected}
  For a sequence set $S$, let $S = \biguplus_{i}S_i$ be the
  decomposition of $S$ into maximal connected components. Let $S'_i$
  be a minimal $\alr$-span of $S_i$. Then
  $S' = \biguplus_{i}S'_i$ is a minimal $\alr$-span of $S$
\end{restatable}

We can show that a minimal $\alr$-span of a connected set is also
connected. By definition, the $\alr$-span is a set of sequences that
has $\alr$, and by the first point in the definition of $\alr$ on
sequence sets, we deduce that there is an $I$ such that all sequences
in the minimal span start with $\act(I)$.

\begin{restatable}{lemma}{spanFirstInfo}\label{lem:span-first-I}
  Let $S$ be a connected set and $S'$ be a minimal $\alr$-span of
  $S$. Then $\exists I \in \Ii$, such that all sequences in $S'$ start
  with $\act(I)$.
\end{restatable}

The next observation says how to find this $I$. 
\begin{restatable}{lemma}{spanFindFirstInfo}\label{lem:span-finding-first-I}
  Let $S$ be a connected set. If there is an $I$ such that every
  sequence of $S$ has an action in $\act(I)$, then there is a minimal
  $\alr$-span $S'$ of $S$ such that all sequences in $S'$ start with
  $\act(I)$.
\end{restatable}

If there is no such $I$, then we need to enumerate over all
information sets to find the smallest. Once we fix an $I$, the next
lemma says how to find a minimal $\alr$-span which starts with
$\act(I)$.

\begin{restatable}{lemma}{spanAfterFixingInfo}\label{lem:span-after-fixing-an-I}
  Let $S$ be a connected set, and let $I \in \Ii$. An $\alr$-span of smallest size among all
    $\alr$-spans starting with $\act(I)$ is the following:
    $S' = \bigcup_{a \in \act(I)}a H_a'$ where $H_a'$ is a minimal
    $\alr$-span of $S_a \cup S_{\bar{I}}$ where $S_{\bar{I}} := \{ s
    \in S \mid s \text{ contains no 
      action from } \act(I)$.
                \end{restatable}

\paragraph*{Algorithm for computing minimal ALR-span.} Based on the
lemmas above, we can design a recursive algorithm to compute a minimal
$\alr$-span for an input sequence set $S$.  
Firstly, if $S$ is disconnected, based on \cref{prop:min-span-disconnected}, $S$ is decomposed
into maximal connected components and a minimal
$\alr$-span is computed for each component. When $S$ is connected:
(1) for each $I \in \Ii$, compute $H_I = \{ s \mid s ~\text{contains no
  action from} \act(I) \}$; (2) if there is some $I$ such that $H_I =
\emptyset$ (Lemma~\ref{lem:span-finding-first-I}), find the smallest $\alr$-span starting with $\act(I)$
using Lemma~\ref{lem:span-after-fixing-an-I}; (3) else, for each $I$,
compute the smallest $\alr$-span starting from $\act(I)$ using
Lemma~\ref{lem:span-after-fixing-an-I} and return the smallest.
We remark that in the algorithm to find the minimal $\alr$-span, $\epsilon$ might appear in intermediate sequence sets. In that case, we just remove $\epsilon$ from the set and continue the algorithm.

Next, we show that there are game-structures whose minimal
$\alr$-spans are exponentially large.

\begin{restatable}{theorem}{spanLowerBound}\label{thm:lower-bound}
  For every $n > 0$, there exists a game structure $\Tt_n$ of size
  $\Oo(n^2)$ such that the size of a minimal $\alr$-span of $\Tt_n$ is
  $\Omega(2^n)$.
\end{restatable}

\paragraph*{Complexity} 
In order to investigate the complexity of computing a minimal
$\alr$-span of a game structure, we consider the following decision
problem and show it to be in $\NP$. We leave open the question of whether
it is $\NP$-hard.  

\textsc{MIN-ALR-SPAN}: Given a game structure $\Tt$ and an integer $k > 0$, is there a game structure $\Tt'$ such that $|\Tt'| \leq k$ and $\Tt'$ is an $\alr$-span of $\Tt$?

\begin{restatable}{theorem}{spanNP}\label{thm:span-NP}
The decision problem MIN-ALR-SPAN is in $\NP$. 
\end{restatable}
\begin{proof}[Proof Sketch]
One can guess a $\Tt'$ of size at most $k$, and also the linear
combinations required for each $\mu \in X(\Tt)$. This is polynomial in
size. For each monomial $\mu \in X(\Tt)$ we verify if the monomials in
$X(\Tt')$ containing all variables from $\mu$ can generate $\mu$. This
can be done efficiently in polynomial-time. \qedhere \end{proof}


\paragraph*{Efficiently solvable classes}
\label{sec:tractable}

We get back to finding efficiently solvable classes of imperfect recall games, this time using a parameter called \emph{shuffle-depth SD}, that is naturally derived from our algorithm for computing a minimal $\alr$-span. 


When $S$ is disconnected with $S = \uplus S_i$ and each $S_i$ connected : define
SD(S) = $\max_i SD(S_i)$. 
When $S$ is connected and $S$ has $\salr$, define SD(S) = 0.
Otherwise define
\[
SD(S) = 1 + \min_{I}\max_{a \in I}\text{SD}(S_a \cup S_{\bar{I}})
\]
For a $\Tt$ with SD $k$, our minimal $\alr$-span algorithm would encounter sets with $\salr$ after recursively
running upto depth $k$. This takes $\Oo(|\Tt|^{k+1})$ time. 

\begin{proposition}
The minimal $\alr$-span of game structure $\Tt$ with SD = $k$ can be computed in time $\Oo(|\Tt|^{k+1})$.
\end{proposition}

\begin{restatable}{corollary}{efficSolvClass}\label{cor:effic-solv-class}
The maxmin value in a one-player game can be computed in $\mathsf{PTIME}$ for games having structures with constant SD.
\end{restatable}

The game \textbf{III} for any $n$ has SD $2$, and using our algorithm one can find equivalent $\alr$-game with just a quadratic blowup in size.  


\section{Two-player games}\label{sec:two-player}
\input{Figures/fig-two-player-shuffle}
For a two-player game structure $\Tt$ and the corresponding sequence set $S$, we can look at the projection of sequences on individual player actions, $S_{\Max}$ and $S_{\Min}$, consider their $\alr$-spans and knit them together. 
Consider the 2-player game in \cref{fig:2-p-shuffle}. Projections of the sequence set in this game w.r.t individual players would correspond to game structures in \cref{fig:2-p-shuffle-b} (for $\Max$) and \cref{fig:shuffle-a}  (for $\Min$). The structure \cref{fig:2-p-shuffle-b}  already has $\pfr$. As seen before, \cref{fig:shuffle-b} is an $\salr$ witness for \cref{fig:shuffle-a}. We can plug in this $\salr$ witness to all leaf nodes of \cref{fig:2-p-shuffle-b} and get the structure \cref{fig:2-p-shuffle-c}. Information sets are maintained across all copies of \cref{fig:2-p-shuffle-c} as shown in the illustration. More generally, if $\Tt'_{\Max}$ is an $\alr$-span of $\Tt_{\Max}$, and $\Tt'_{\Min}$ is an $\alr$-span of $\Tt_{\Min}$, we can obtain an $(\alr, \alr)$ structure where all nodes of $\Max$ precede all nodes of $\Min$, and a copy of $\Tt'_{\Min}$ is attached to each leaf node of $\Max$, and information sets of $\Min$ are maintained across all copies of $\Tt'_{\Min}$, as shown in \cref{fig:2-p-shuffle}.  
We can assign suitable payoffs in a similar manner to \cref{prop:span-to-game} to get an equivalent game on this new game structure.

\begin{restatable}{theorem}{twopPfrNam}\label{thm:2p-pfr-name}
Solving an $(\nam, \nam)$ game on structure $\Tt$ can be reduced to solving an $(\alr,\alr)$ game on a structure of size $|\Tt'_{\Max}||\Tt'_{\Min}|$ where $\Tt'_{\Max}$ and $\Tt'_{\Min}$ are minimal $\alr$-spans of $\Tt_{\Min}$.
\end{restatable}
Since $(\pfr, \alr)$-games can be solved in polynomial-time, \cref{thm:2p-pfr-name} and Corollary~\ref{cor:effic-solv-class} lead to new polynomial-time solvable classes.
\begin{corollary}\label{cor:2-effic-solv-class}
The maxmin value in a $(\pfr, \nam)$ game where SD of $\Tt_{\Min}$ is constant can be computed in polynomial time. As a conseqeunce, $(\pfr, \salr)$ games can be solved in polynomial time.
\end{corollary}


\section{Conclusion}

We have presented a study of imperfect recall without
absentmindedness, through the lens of A-loss recall. Specifically, we
have given two methods to transform imperfect recall games to A-loss
recall games. Behavioral strategies for the original game can be
obtained by analyzing the transformed game. This investigation has
resulted in new polynomial-time solvable classes of one-player and
two-player games. We have also shown how to find a transformation
of minimal size. It would be interesting to see the influence of these
notions of $\salr$ and $\alr$-span, and the idea of using sequence
sets instead of games, in algorithms that use imperfect recall
abstractions.

In summary, in this work, we have laid the foundations to simplify imperfect
recall in terms of $\alr$. We do hope that this perspective leads to
further theoretical and experimental investigations.


\bibliographystyle{plainurl}

\bibliography{main}

\appendix

\setcounter{proposition}{0}
\setcounter{lemma}{0}
\setcounter{theorem}{0}
\setcounter{corollary}{0}

\renewcommand{\thetheorem}{T\arabic{theorem}}
\renewcommand{\theproposition}{P\arabic{proposition}}
\renewcommand{\thelemma}{L\arabic{lemma}}
\renewcommand{\thecorollary}{C\arabic{corollary}}


~\\~\\
\begin{center}
{\huge{Appendix}}
\end{center}

\section{Shuffled A-loss recall}

\shuffleSameLeafMonomials*
\begin{proof}
Let $L$ and $L'$ be set of leaf nodes of $\Tt$ and $\Tt'$ respectively. From definition of $\salr$, there is a bijection $f$ between the sets $\Hh(L)$ and $\Hh(L')$ such that for $s \in \Hh(L)$, $f(s)$ is a permutation of $s$. Since $\forall s \in \Hh(L),\mu(s) = \mu(f(s))$ it follows that $X(\Tt) = X(\Tt')$. 
\end{proof}

\alrSetStructEquiv*
\begin{proof}
First, starting with an $\alr$-game structure $\Tt$, we will show by induction on the size of $\Tt$, that $\Hh(L)$ has $\alr$. 

When $\Tt$ is a trivial leaf node, we have $\Hh(L) = \{ \epsilon \}$ which has $\alr$ by definition. Suppose the statement holds for all game structure of size at most $k$. 
Let $\Tt$ be a game structure of size $k+1$. First note, that any structure can be reduced to a structure where $\chance$ nodes do not have $\chance$ children, and with the same leaf histories $\Hh(L)$. This can be done by absorbing the $\chance$ children into the parent.

Now we can split into several cases and sub-cases.

Case 1: The root node $r$ of $\Tt$ is a $\chance$ node. When $r$ is a $\chance$ node, there can be two cases.

Case 1a: All children of $r$ (which are player nodes) are in one single information set $I$. In this case, one can construct a structure $\Tt'$ with first two levels of $\Tt$ swapped i.e. the root of $\Tt'$ is a player node in information set $I$ and each actions leading to $\chance$ nodes which again leads to the substructure from third level of $\Tt$. This way, $\Tt$ is reduced to a structure $\Tt'$ with player node at root and $\Hh(\Tt) = \Hh(\Tt')$. This falls in Case 2.

Case 1b: Not all children of $r$ are in single information sets. From definition of $\alr$, substructure rooted at two children of $r$ in two distinct information sets cannot share an information set. This implies, these are substructures of size at most $k$ and by inductive hypothesis, their set of leaf histories have $\alr$. Now since their leaf histories are mutually disconnected, this implies that $\Hh(L)$ has $\alr$.

Case 2: The root node $r$ is a player node.
In this case, for any action $a$ out of root node, and any two node in the sub-tree out of $a$, $a$ is the common prefix of histories of both node. Hence from definition of $\alr$, $\Tt$ has $\alr$ would imply for any $a$, the substructure rooted at each children must individually have $\alr$. By inductive hypothesis the sub-structure has size less than $k+1$ and hence the leaf histories of substructure has $\alr$ as well. Since $\Hh(L)$ is connected in this case, it follows that $\Hh(L)$ has $\alr$ as well. 

~\\
~\\ 
The other direction can be proved similarly, again by using induction on the size of $\Tt$ and breaking into similar cases.  
\end{proof}
\salrDisc*
\begin{proof}
Since $S$ is disconnected from definition of $\salr$ it follows that when $S$ has $\salr$ with $\salr$ witness $S'$, $S'$ is also disconnected, i.e. the maximal connected components of $S'$ are $\salr$ witnesses of maxmimal connected components of $S$. Hence each $S_i$ has $\salr$.

The other direction follows since the union of the $\salr$ witnesses for connected components of $S$ is an $\salr$ witness of $S$. 
\end{proof}
\commonact*
\begin{proof}
For the forward direction, let $S'$ be an $\salr$ witness of $S$. Since $S$ is connected, it follows that $S'$ is connected as well. Now from definition of $\alr$ sets, $\exists I$, such that all sequences in $S'$ starts with actions from $\act(I)$. This proves the first statement. 
Now again it follows from the definition of $\alr$ sets that for each $a \in \act(I)$, $S'_a$ has $\alr$. Moreover $S'_a$ is $\salr$ witness of $S_a$. This proves the second statement. 

For the other direction $\bigcup_a aS'_a$ will have $\alr$ by definition and hence is an $\salr$ witness of $S$, implying $S$ has $\salr$.   
\end{proof}
\salrPtime*
\begin{proof}
Suppose $S$ has $\salr$ and $I$ be an information set such that all sequences in $S$ has actions from $\act(I)$. Let $S'$ be an $\salr$ witness of $S$. We claim that for each $a \in \act(I)$, $S'_a$ will have $\alr$. 

For this we provide an equivalent formulation of $\alr$ sets. A set has $\alr$ iff every subset of size $2$ has $\alr$. This is true because $\Tt$ has $\alr$ iff every pair of leaf histories $s_1$ and $s_2$ of $\Tt$ satisfy the constraints imposed by $\alr$ condition, i.e. either (i) $s_1 = sa_1s'_1,~ s_2 = sa_2s'_2$ for $a_1 \neq a_2$ with $a_1,a_2 \in \act(I')$ for some $I'$ or (ii)  $s_1 = ss'_1,~ s_2 = ss'_2$ where $s'_1$, $s'_2$ are disconnected. Now it follows that if $s_1$ and $s_2$ share a common action $a$, the pair of sequences on removing this common $a$ from both will still satisfy the $\alr$ conditions. This implies that $S'_a$ will also have $\alr$ for each $a$. As a result for each $a$, $S_a$ will have $\salr$.

For the other direction $\bigcup_a aS'_a$ will have $\alr$ by definition and hence is an $\salr$ witness of $S$.  
\end{proof}
\onepShufflePtime*
\begin{proof}
Since $\salr$ witnesses are specific kinds of $\alr$-spans, using the construction in the proof of \cref{prop:span-to-game}, we can reduce the input game $G$ to a new game $G'$ with $\alr$ and with the same payoff polynomial. Since game structures with $\salr$ have SD $0$, this follows from \cref{cor:effic-solv-class}. 
\end{proof}
\section{Span}

\subsection{Existence of $\alr$-span}

\input{Figures/fig-app-span}
\alrspan*

\begin{proof} Suppose $\Tt$ has information sets $\Ii = \{ I_1, \dots, \I_n \}$ and actions $A$. We will provide an explicit $\alr$-span  $\Tt'$ over the same information sets $\Ii$ and actions $A$.
An example of this construction for $n=3$ can be found in \cref{fig:app-alossSpan}, where the game structure in \cref{fig:app-alossSpan-c} is an $\alr$-span of the structure in \cref{fig:app-alossSpan-a} obtained in this manner. 
 
Let $S = \{\epsilon\} \cup  \{a_{i_1}\dots a_{i_k} \mid k \leq n, \forall j~ a_{i_j} \in \act(I_j) \}$. $\Tt'$ be the game structure over the vertex set $V' = \{u_s | s \in S\}$, with $v_{\epsilon}$ being the root node and the set of leaves $L' = \{ v_s \mid \lvert s \rvert = n \}$. 
The set of actions is $A$ where for $u_s \in V'\setm L'$, $u_s \xra{a} u_{sa}$. The set of information sets is $\Ii$ and for each $ 1 \leq i \leq n, I_i = \{ u_s \mid  \lvert s \rvert = i-1\}$.

We claim that $\Tt'$ is an $\alr$-span of $\Tt$. Firstly $\Tt'$ has $\alr$ because it follows from the recursive definition of $\alr$-sequence sets that $S$ has $\alr$.
Now let $\mu(s)$ be a monomial generated by some leaf sequence $s$ in $\Tt$. We will show that the leaf monomials of the set $\{ s' \in \Hh(L') \mid \act(s) \incl \act(s') \}$ generate the monomial $\mu(s)$ where $\act(s)$ is the set of actions in $s$. More particularly the expression $\sum\limits_{s' \in \Hh(L') \mid \act(s) \incl \act(s')} \mu(s')$ reduces to $\mu(s)$. This is because the previous expression can be re-written as

 $\mu(s)\prod\limits_{I \mid \act(I) \cap \act(s) = \emps}(\sum_{a \in \act(I)} x_a)$ which reduces to $\mu(s)$ by repeated applications of $\sum_{a \in \act{I}}x_a = 1$ for $I$ with $ \act(I) \cap \act(s) = \emps$. Hence $\Tt'$ is an $\alr$-span of $\Tt$.

~\\ 
\end{proof}

\spanToGame*

\begin{proof}
We will provide $\d'$ and $\Uu'$ to construct the game $G'$ on $\Tt'$. 

Firstly, we choose $\d'$ which at every $\chance$ node assigns an uniform distribution on its outgoing edges(any distribution with full support would work). We will factor out these $\chance$ probabilities later in the payoff function $\Uu'$.   
For a leaf $t \in L$, it follows from definition of span that a linear combination of $\mu(t')$ for $t' \in \Tt'$, reduces to $\mu(t)$. Let that combination be $\sum_{t' \in L'}c^t_{t'}\mu(t')$. Recall that $\prob_{\chance}(t')$ is the product of $\chance$ probabilities by $\d'$ on path to $t'$ in $\Tt'$.  Define $\Uu'(t') = \frac{\sum_{t \in L}\prob_{\chance}(t)c^t_{t'}\Uu(t)}{\prob_{\chance}(t')} $ and let $G'$ be the resulting game. Observe that payoff polynomial of $G'$ turns out to be 
$\sum\limits_{t' \in L'}\mu(t')\sum_{t \in L}\prob_{\chance}(t)c^t_{t'}\Uu(t)$. By rearranging terms, this is equivalent to $\sum\limits_{t \in L} (\prob_{\chance}(t)\Uu(t))\sum_{t' \in L'}c^t_{t'}\mu(t')$. It follows that this reduces to the payoff polynomial of $G$.
\end{proof}

\subsection{Minimal $\alr$-span }

\minSpanDisc*

\begin{proof}
Firstly we can see that $S'$ is indeed an $\alr$-span of $S$. Also for any distinct $i$ and $j$, since $S_i$ and $S_j$ are not connected, it follows that $S'_i$ and $S'j$ are not connected as well. 

Let us assume that there is an $\alr$-span $\hat{S}$ of $S$ such that $ \mid  \hat{S} \mid   <  \mid  S' \mid   = \sum\limits_i  \mid  S_i' \mid  $. For each $i$, let $\hat{S}_i$ be a minimal subset of $\hat{S}$ that spans $S_i$. Again from the similar previous argument, for any distinct $i$ and $j$, $\hat{S}_i$ and $\hat{S}_j$ are not connected. Since $ \mid  \hat{S} \mid  = \sum\limits_i  \mid  \hat{S}_i \mid  $, $\exists i$ such that $ \mid\hat{S}_i\mid   <  \mid S_i'\mid  $. But this  contradicts the minimality of $S'_i$ for each $S_i$. This completes the proof.  

\end{proof}

\spanFirstInfo*

\begin{proof}
Given a connected sequence set $S$, any minimal $\alr$-span $S'$ of $S$ is also connected. This is because if one could find a decomposition $S' = \uplus_i S_i'$, such that $S_i'$'s are disconnected, then the subsets of $S$ individually spanned by each $S'_i$ would be disconnected as well leading to a contradiction.

Since an $\alr$-span is a set with $\alr$, it follows from definition of $\alr$-sets that for some $I$, all sequences in $S'$ must start with an action from $\act(I)$. 
\end{proof}

\spanAfterFixingInfo*

\begin{proof}
Suppose there is an $\alr$-span $\hat{S}$ of $S$ starting with $\act(I)$ such that $\lvert \hat{S} \rvert < \lvert S' \rvert$. 
Observe that $\hat{S} = \uplus_{a \in \act(I)}a\hat{S}_a$. 
We will show that for each $a$, $\hat{S}_a$ is an $\alr$-span of $S_a \cup S_{\bar{I}}$. Since $\lvert \hat{S} \rvert = \sum_{a} \lvert \hat{S}_a \rvert$ and $\lvert S' \rvert = \sum_{a} \lvert H'_a \rvert$ this would contradict the minimality of some $H'_a$.

Since $\hat{S}$ is a span of $S$, for any $s \in S_a$, $\mu(as) = x_a\mu(s)$ is obtained by reduction from $\sum_{a \in \act(I)}x_af_a$ where $f_a$ is combination of monomials of sequences with $a$. Substituting $x_a$ with $0$ for every term containing $x_a$ we get an expression which reduces to $0$. This implies the expression $f_a$ reduces to $\mu(s)$ itself. Again for $s \in S_{\bar{I}}$ substituting for each $a \in \act(I)$, $x_a$ by $1$ we get an expression devoid of $x_a$ for $a \in \act(I)$ that reduces to $\mu(s)$. 
These together implies that $\hat{S}_a$ is an $\alr$-span of $S_a \cup S_{\bar{I}}$.
This completes the proof.  
\end{proof}

\spanFindFirstInfo*

\begin{proof}
We can show this by induction on the size of $\Ii$. This is trivially true when $\lvert \Ii \rvert = 1$. Let this be true when number of information sets is at most $k$. Now if we have a sequence set $S$ over $k+1$ information sets, since $S$ is connected it follows from \cref{lem:span-first-I} that in a minimal span $S'$ all sequences start with $\act(I')$ for some $I'$. Moreover for each $a \in \act(I')$, it follows from \cref{lem:span-after-fixing-an-I} that $S'_a$ is a minimal $\alr$ span of $S_a \cup S_{\bar{I'}}$. Now by inductive hypothesis since every $S_a \cup S_{\bar{I'}}$ has actions from $I$, we can assume that $S'_a$ starts with actions from $I$. Finally one can swap the order of actions in each sequence, so that actions from $I$ appear at the start. The resulting set will still have $\alr$ because every pair of sequences satisfy the conditions put down in the proof of \cref{lem:salr-common-act}. Since the resulting set still remains a span of $S$, this completes the proof.  
\end{proof}

We provide the algorithm for computing minimal $\alr$-span described in the main text. 
\begin{algorithm}[H]
\caption{Compute minimal $\alr$-span}
\label{algo:min-span}
\begin{algorithmic}[1]
\STATE \textbf{Input} : $S$ 
\STATE \textbf{Output}: minimal $\alr$-span $S'$ of $S$ 
\IF{$S$ is connected} 
\IF{$\exists I$ such that every $s \in S$ contains an action from
  $\act(I)$ \label{algo:min-span-connected}} \FOR {$a \in \act(I)$}
\STATE $S_a' \gets $ minimal $\alr$-span of $S_a$\ENDFOR
\RETURN $\bigcup_{a} aS_a'$ 
\ELSE 
\FOR {$I \in \Ii$}
\STATE $S_{\bar{I}} \gets \{ s \in S \mid s \text{ contains no action from } \act(I) \}$
\FOR {$a \in \act(I)$}
\STATE $S_a' \gets $ minimal $\alr$-span of $S_a \cup S_{\bar{I}}$
\ENDFOR
\STATE $\hat{S}_I \gets \bigcup_{a \in \act(I)} S'_a$
\ENDFOR
\RETURN $\hat{S}_I$ that satsfies $\min_{I \in \Ii}  \lvert \hat{S}_I \rvert $ \label{algo:min-span-all}
\ENDIF
\ELSE
\STATE $S = \uplus S_i$ where each $S_i$ is connected 
\STATE $S'_i \gets $ minimal $\alr$-span of $S_i$ \label{algo:min-span-disconnected}
\RETURN $\cup_i S'_i$
\ENDIF
\end{algorithmic} 
\end{algorithm}

As mentioned in the main text, at every step if $\epsilon$ is also in the set we get rid of it. This doesn't affect the computation of the minimal $\alr$-span. This is because for any $\alr$ structure $\Tt$, $X(\Tt)$ can always generate $\epsilon$. This can be show using the existence of ``strongly branching'' (see \cref{lem:strong-branch}) subsets in $\Hh(L_{\Tt})$. In \cref{algo:min-span}, at every recursive call for computing minimal $\alr$-span, the returned structure will have $\alr$ and hence will generate $\epsilon$ anyway. 

\spanLowerBound*

\begin{proof}
  For $n > 0$, let $\Ii_n = \{I_1,\dots,\I_n\}$, $\act(I_i) = \{a^i,b^i\}$ and $A_n = \bigcup_i \act(I_i)$. Define $S_n = \{a^ib^j \mid i < j\} \cup \{b^ia^j \mid i < j\} \cup \{a^ia^j \mid i < j\} \cup \{b^ib^j \mid i < j\} $. Let $\hat{S}_n =  A_n \cup
S_n$. We claim that the class of game structures with sequence set $\hat{S}_n$ has minimal $\alr$-span of size $\Omega(2^n)$.

  Since $A_n \cup S_n$ is connected, according to \cref{lem:span-first-I}, a minimal span must start with actions from $\act(I)$ for some $I$. Now in order to choose an $I$ and apply \cref{lem:span-after-fixing-an-I}, 
we will observe that the set of monomials of $S_n$ is symmetric w.r.t each $I$ and the choice of $I$ doesn't make a difference. For a set $S$, observe that for any $I \in \Ii_n$, the leaf monomials of the set $\{s \in A_n \cup S_n  \mid  \act(I) \cap \act(s) = \emptyset\}$ and that of the set $S_{n-1} \cup A_{n-1}$ are same up to renaming of variables. 
Also, for any $I$ and for any $a \in \act(I)$, the leaf monomials of the set $\{ ss'  \mid  sas' \in A_n \cup S_n \} $ and the set $A_{n-1} \cup \{ \epsilon \}$ are again same upto renaming of variables. 
Hence for any choice of $I$ from $\Ii_n$, applying \cref{lem:span-after-fixing-an-I} yields a minimal $\alr$-span.

Following this argument if we pick $I_n$, since $(S_{n-1} \cup A_{n-1}) \cup (A_{n-1} \cup \{ \epsilon \}) = \hat{S}_{n-1}$ (ignoring the $\epsilon$) we are then left with finding the minimal $\alr$-span of $\hat{S}_{n-1}$. 
Let $T(n)$ be the size of a minimal $\alr$-span of $\hat{S}_n$. Then from \cref{lem:span-after-fixing-an-I} it follows $T(n) = 2T(n-1)$. 
Since for $n=1$, the minimal $\alr$-span has size $2$, it follows that a minimal $\alr$-span of $\hat{S}_n$ has size $\Omega(2^n)$.

\end{proof}

\spanNP*

To prove \cref{thm:span-NP} we first define something called a \emph{strongly branching set} recursively.
The trivial set $S = \{ \epsilon \}$ is a strongly branching set. $S$ is a strongly branching set if there exists $I \in \Ii$ such that there is a partition of $S = \uplus_{a_i \in \act(I)} a_iS_i $ such that (i) $a_iS_i$s are non-empty (ii) each of $S_i$ are strongly branching sets.

We prove the following lemma, that we will use in our proof. 
\begin{lemma}\label{lem:strong-branch}
Let $S$ be a sequence set with $\alr$. Then $\sum\limits_{s \in S}\mu(s) = 1$  under strategy constraints iff $S$ is a strongly branching set.
\end{lemma}
\begin{proof}
First we will prove the backward direction by induction on the number of actions. Observe that any strongly branching set always has $\alr$, hence $\alr$ assumption is redundant in this case.  When $S = \{ \epsilon \}$ recall that $\mu( \epsilon ) = 1$. Now for any non-trivial strongly branching set $S$, from definition we have $S = \uplus_{a_i \in \act(I)} a_iS_i$ where $S_i$s are strongly branching. We have $\sum\limits_{s \in S}\mu(s) = \sum_{a_i} x_{a_i}(\sum\limits_{s' \in S_i}\mu(s'))$. By induction hypothesis the statement holds for each $S_i$ and hence this equals 1. Hence this extends to $S$ as well. 

For the the forward direction we will again use induction on the number of actions. The statement is true when $A = \emps$. Now suppose $S$ be a sequence set over non-trivial action set.

For any set $S'$, if $\sum\limits_{s \in S'}\mu(s)$ is a constant value then this value is at least 1 due to the fact that any $\mu(s)$ has either no constant terms or 1 in its full expansion.
Observe that $S$ is connected.Otherwise, if $S$ has at least two connected components $S'$ and $S''$, since $\mu(S')$ and $\mu(S'')$ have no variable in common, $\sum\limits_{s \in S}\mu(s)$ would become strictly more than 1.

 Since $S$ is connected and has $\alr$, it follows that $S= \uplus_{a_i \in \act(I)} a_iS_i $ for some $I$ where each $S_i$ has $\alr$. Let $Z_i = \sum\limits_{s' \in S_i}\mu(s') $. Then we have $\sum\limits_{s \in S}\mu(s) = \sum_{a_i} x_{a_i}Z_i = 1$, which implies each of $Z_i$s must equal 1. Hence by induction hypothesis each of $S_i$s are strongly branching sets and as a result $S$ is also a strongly branching set. 

\end{proof}
\begin{proof}[Proof of \cref{thm:span-NP}]
Given a structure $\Tt$, we will guess an $\alr$-span $\Tt'$ of size at most $k$ along with the linear combinations for each leaf monomial. Since from \cref{algo:min-span} we can get the minimal span and all the co-efficient needed lies in $\{0,1\}$, we can restrict to this kind of combinations. Hence, this is of size at most $k|\Tt|$.  

Firstly, checking if $\Tt'$ has $\alr$ can be done efficiently by traversing the tree from root.
Now, to verify if each for $s \in \Hh(L)$, the linear combination for generating $\mu(s)$, is correct, we just need to consider all $s' \in \Hh(L')$ that has all actions from $s$,
i.e. we consider the set $S_s = \{ s' \in \Hh(L') \mid \act(s) \incl \act(s') \}$ where $\act(s)$ is the set of actions in $s$. This is because monomials of sequences not containing all actions form $\act(s)$, will either vanish (thus not affecting the final outcome) or not reduce to $\mu(s)$.

 Hence, if $\sum \mu'$ is the candidate linear combination for $\mu$,  $\sum \mu'$ reduces to $\mu$ iff the $\sum \frac{\mu'}{\mu} = 1$. Let $\frac{s'}{s}$ denote the sequence $s'$ restricted to actions from $\act(s')\setminus \act(s)$. But since $S_s$ is assumed to have $\alr$, this set will also have $\alr$. Hence from \cref{lem:strong-branch}, we need to check if the set $\{\frac{s'}{s} | s' \in S_s\}$ has a strongly branching subset. This can be checked recursively in polynomial time.  
\end{proof}

\efficSolvClass*

\begin{proof}
Since $\Tt$ has SD $k$, the size of the minimal $\alr$-span $\Tt'$ of $\Tt$ is at most $|\Tt|^{k+1}$. Now given a game $G = (\Tt, \d, \Uu)$ on $\Tt$, consider the game $G' = (\Tt', \d', \Uu')$ obtained using the construction in \cref{prop:span-to-game}. It follows from \cref{prop:span-to-game} that solving $G$ can be reduced to solving $G'$. Since $\d'$ assigns uniform distribution at every chance nodes, probabilities in $\d'$ need at most $\Oo(\log(|\Tt|^{k+1}))$ bits. It also follows from \cref{algo:min-span} that the coefficients $c^t_{t'}$ in the proof of \cref{prop:span-to-game} can be either $1$ or $0$ for the minimal $\alr$-span. Hence each payoff in $\Uu'$ uses $\Oo(\log{(U|\Tt|^{k+1})})$ bits where $U$ is the maximum possible payoff in $\Uu$.
Since $k$ is constant, hence size of $G'$ is polynomial in the size of $G$. 
\end{proof}

\section{Extension to Two-player games}

\begin{lemma}
Let $\Tt$ be a two player game structure with sequence set $S$ and let $S_{\Max}$ and $S_{\Max}$ be the projection of $S$ onto actions of $\Max$ and $\Min$ respectively. Then, there exists game structures $\Tt_{\Min}$, $\Tt_{\Min}$ such that $S_{\Max}$(resp $S_{\Min}$) is the sequence set of $\Tt_{\Max}$(resp. $\Tt_{\Min}$).
\end{lemma}
\begin{proof}
Fix a behavioral strategy for $\Min$ and this makes the structure a one-player structure with $\Max$ player. The game structure of this game, $\Tt_{\Max}$ is a structure such that the sequence set of $\Tt_{\Max}$ is $S_{\Max}$. 
This works for $\Min$ in a similar way. 
\end{proof}
\twopPfrNam*
\begin{proof}
Given $\Tt_{\Max}'$ and $\Tt_{\Min}'$, the $\alr$ spans of $\Tt_{\Max}$ and $\Tt_{\Min}$, we construct $\Tt'$ as follows : 
the top part of $\Tt'$ is the game structure $\Tt_{\Max}'$ and to each leaf of $\Tt_{\Max}'$ a copy of $\Tt_{\Min}'$ is attached. Nodes in two different copies of $\Tt_{\Min}'$ that are in the same information set in $\Tt_{\Min}'$ are also in same information set of $\Tt'$. 

Now $\Tt'$ is a structure of size $\lvert \Tt_{\Max}' \rvert \lvert \Tt_{\Min}' \rvert$. We will show that $\Tt'$ spans $\Tt$, i.e. it satisfies the 2nd condition in the definition of $\alr$ span. Once we have shown this, the proof will follow using the same construction in \cref{prop:span-to-game}.

Any leaf history in $\Tt'$ will have a $\Max$ part followed by a $\Min$ part, i.e. for $t' \in L'$, $\mu(t') = \mu_{\Max}(t')\mu_{\Min}(t')$. This is true for $\Tt$ as well. Since $\Tt_{\Max}'$ and $\Tt_{\Min}'$ are $\alr$ spans of $\Tt_{\Max}$ and $\Tt_{\Min}$, for any $t \in L$, for $\mu_{\Max}(t)$ we have witness linear combination $\sum_{t'}c^t_{t'}\mu_{\Max}(t')$. The same holds for $\Min$. Hence we have $\mu(t) = \mu_{\Max}(t)\mu_{\Min}(t) = f_{\Max} f_{\Min}$ where $f_i$ is linear combination of monomials from $\Tt'_i$. Expanding this, it follows that $\sum_{t' \in L'}c^t_{t'}\mu(t')$ reduces to any $\mu(t)$. This completes the proof.     
\end{proof}


\end{document}